\documentclass[]{article}

\usepackage{moreverb,url}

\usepackage{amsmath,amsfonts,amsthm} % Math packages
\usepackage{natbib}

\usepackage{hyperref}
\usepackage[pdftex]{graphicx}	
\usepackage{url}
\usepackage[]{natbib}
\usepackage{float}
\usepackage{xcolor}
\usepackage[margin=1in]{geometry}

\usepackage{xr}
\hypersetup{
	colorlinks,
	citecolor=black,
	filecolor=blue,
	linkcolor=blue,
	urlcolor=blue
}

\newcommand{\bs}[1]{\boldsymbol{#1}}
\newcommand{\e}{\varepsilon}
\newcommand{\expec}{\mathbb{E}}
\newcommand{\pr}{{P}}
\newcommand{\wh}[1]{\widehat{#1}}
\newcommand{\expit}{\operatorname{expit}}
\newcommand{\logit}{\operatorname{logit}}

\newtheorem{theorem}{Theorem}
\newtheorem{corollary}{Corollary}
\newtheorem{lemma}{Lemma}

\newtheorem{note}{Comment}
\newtheorem{condition}{Condition}
\newtheorem{definition}{Definition}

\newtheorem{innercustomthm}{Lemma}
\newenvironment{customthm}[1]
{\renewcommand\theinnercustomthm{#1}\innercustomthm}
{\endinnercustomthm}

\begin{document}

\title{A Targeted Approach to Confounder Selection for High-Dimensional Data}

\author{Asad Haris\footnote{Department of Earth, Ocean and Atmospheric Sciences, University of British Columbia. \emph{email}: aharis@eoas.ubc.ca} \ and\  Robert Platt\footnote{Department of Epidemiology, Biostatistics and Occupational Health, McGill University. \emph{email}: robert.platt@mcgill.ca}}

%\corrauth{Asad Haris, 
%Epidemiology, Biostatistics and Occupational Health, McGill University, Montreal QC, Canada.}
%\email{asad.haris@mail.mcgill.ca}

% Keywords command
\providecommand{\keywords}[1]
{
	\small	
	\textbf{\textit{Keywords---}} #1
}

\maketitle

\begin{abstract}
	We consider the problem of selecting confounders for adjustment from a potentially large set of covariates, when estimating a causal effect. Recently, the high-dimensional Propensity Score (hdPS) method was developed for this task; hdPS ranks potential confounders by estimating an importance score for each variable and selects the top few variables. However, this ranking procedure is limited: it requires all variables to be binary. We propose an extension of the hdPS to general types of response and confounder variables. We further develop a group importance score, allowing us to rank groups of potential confounders. The main challenge is that our parameter requires either the propensity score or response model; both vulnerable to model misspecification. We propose a targeted maximum likelihood estimator (TMLE) which allows the use of nonparametric, machine learning tools for fitting these intermediate models. We establish asymptotic normality of our estimator, which consequently allows constructing confidence intervals. We complement our work with numerical studies on simulated and real data.
\end{abstract}


\section{Introduction}
\label{sec:introduction}

The goal of most epidemiological studies is to determine the relationship between some exposure and an outcome variable. Ideally, a randomized controlled trial (RCT) is used to assess this relationship where subjects are randomized to treatment or placebo. This randomization ensures that, on average, subject characteristics are the same in both treatment arms. However, due to ethical or financial constraints, a RCT is not always feasible and investigators often use observational studies. In this case, analyzing the exposure-response relation must account for the additional factors, namely confounders~\citep{greenland2001confounding}. For decades, it has been known that not adjusting for confounders or presence of unmeasured confounders can lead to data exhibiting an exposure-response association when none exists or even reverse the direction of true effect~\citep{neyman1938lectures,bickel1975sex}. Confounding bias is a particular concern for epidemiological studies~\citep[see e.g.][]{macmahon2001reliable,rothman2015modern}. While confounder adjustment is important for bias reduction, adjusting for irrelevant variables can lead to efficiency losses~\citep{greenland2008invited,schisterman2009overadjustment,rotnitzky2010note,myers2011effects,patrick2011implications}. Thus, it is crucial to correctly identify the  set of confounders. 

%One approach to adjusting for confounders is building a multivariate model of the response as a function of the exposure and confounders. Alternative methods use instead the  propensity score, the propensity score is the probability of a subject being exposed based on their confounding variables. %\cite{haukoos2015propensity} outlined four general ways propensity scores are used, 1. propensity score matching, 2. stratification on the propensity score, 3. covariate adjustment using the propensity score and, 4. inverse probability of treatment weighting by the propensity score. 
%An accurate estimate of the multivariate model or propensity score relies on correctly identifying an appropriate set of confounders. 

Traditionally, confounders are identified from a collection of variables. Each variable in the collection, is classified as either a confounder or not based on expert knowledge/assumptions of the underlying data generating mechanism~\citep{robins2001data}. This approach is suitable if the true confounder set is a subset of our collection; this is often not true and rarely verifiable. With the advent of economical data collection and storage technology, high-dimensional data have become increasingly accessible. In high-dimensional data we often have thousands of variables for each subject and, it is reasonable to assume that a subset of these variables is equal to, or a proxy, for the set of confounders. However, manually sifting through thousands of variables is not feasible, and furthermore,  the underlying mechanisms for each variable may not be well understood. Thus, we require an automated procedure which can efficiently select a suitable subset of confounders.

The problem of variable selection for building precise prediction models has been extensively studied in the literature for both low and high-dimensional data. However, the literature for confounder selection, particularly for high-dimensional data, is sparse. A recent proposal is the C-TMLE \citep{laan2010collaborative,gruber2010application,laan2011targeted}, which performs confounder selection when used for estimating the average treatment effect (ATE); a computationally faster variation, the SC-TMLE was proposed by \cite{ju2019scalable}. Another proposal is the outcome adaptive lasso~\citep{shortreed2017outcome}, an application of the adaptive lasso~\citep{zou2006adaptive} with weights informed by the outcome regression model. These methods are designed for confounder selection but not confounder ranking, which might be of interest to investigators. 

One proposal for confounder ranking is the variable importance measure (VIM) of \cite{chambaz2012estimation}, while theoretically appealing, it is computationally cumbersome for even moderate-dimensional data. 
Another approach, popular in pharmacoepidemiology, is the hdPS proposal of \cite{schneeweiss2009high}. It has been used for various data analyses~\citep{brookhart2010confounding,patorno2010anticonvulsant,rassen2010cardiovascular,rassen2012using,schneeweiss2010comparative,schneeweiss2010variation,huybrechts2011risk,toh2011confounding,he2014mining,pang2016effect}, its performance has been studied empirically, and some extensions/modifications have been proposed~\citep{patrick2011implications,rassen2011covariate,franklin2014plasmode,franklin2015regularized,franklin2017comparing,guertin2016head,schneeweiss2017variable,enders2018potential,karim2018can,ju2019propensity}. Despite its popularity in pharmacoepidemiology, hdPS suffers from drawbacks: primarily, it discretizes each variable into three binary variables and ranks the resulting binary variables. Interpreting the ranking % or selection o
of these artificial variables is difficult. Furthermore, hdPS cannot rank groups of variables, which might be of interest to investigators.%, e.g. variables clustered by location or batch effects.

%a major concern is that hdPS discretizes each variable into three binary variables and ranks the resulting binary variables. Interpreting the ranking or selection of these created binary variables is thus infeasible. Furthermore, hdPS cannot rank groups of variables, which might be of interest to investigators.

%Furthermore, investigators might also be interested in ranking groups of variables, e.g. variables clustered by location or batch effects, this cannot be achieved by hdPS or other methods.

In this paper, we propose a novel confounder ranking technique which easily scales to high-dimensional data. We propose a generalization of the Bross score~\citep{bross1966spurious} as used by \cite{schneeweiss2009high} for hdPS. %Beyond generalizing the Bross formula, this paper makes important novel contributions to the literature which we summarize next. 
With the generalized Bross formula, we propose two ranking schemes which require only a binary exposure, placing no restrictions on other variables. % (possible extensions to general exposure variables is also discussed). 
Our framework naturally extends to ranking groups of variables. We establish asymptotic normality of our estimators allowing us to calculate confidence intervals and $p$-values.% for hypothesis testing. 

%In Section~\ref{sec:methods}, we present our novel ranking schemes for individual and groups of covariates. In Section~\ref{sec:estimation}, we present two naive plug-in estimators, followed by an efficient, doubly robust one step correction. We then develop a TMLE and establish asymptotic normality. In Sections~\ref{sec:simulations} and \ref{sec:bostonData}, we present numerical experiments and data analysis. Finally, concluding remarks and future work is discussed in Section~\ref{sec:conclusion}. 
%In Section~\ref{sec:extensions} we discuss extensions of our method, its limitations and potential solutions to these limitations. We present simulation studies in Section~\ref{sec:simulation} and analyze the Boston housing data and the \textcolor{red}{[some high dimensional data]} in Section~\ref{sec:datanalysis}. Finally, concluding remarks and future work is discussed in Section~\ref{sec:conclusion}.  

\section{A Novel Confounder Scoring Scheme}
\label{sec:methods}

\subsection{Motivation}
\label{sec:motivation}
To motivate our proposal, we discuss desirable properties of a confounder ranking system and, gaps in the existing literature. We then present the statistical intuition behind the Bross formula %~\citep{bross1966spurious} 
and consequently our extension of it. 

We require, a (semi-)automated confounder ranking scheme which can be extended to groups of variables, handle all variable types and, can be efficiently computed for high-dimensional data. The C-TMLE %~\citep{laan2010collaborative,gruber2010application} 
and its scalable extension, the SC-TMLE, %~\citep{ju2019collaborative}), 
do not rank variables and cannot perform grouped variable selection. While SC-TMLE is substantially faster than C-TMLE, existing software~\citep{ju2017ctmle} is still computationally cumbersome for high-dimensional data. Similary, the outcome adaptive lasso~\citep{shortreed2017outcome} does not rank variables; additionally, it assumes restrictive generalized linear models for the outcome and propensity score models. The VIM proposal of \cite{chambaz2012estimation} is theoretically appealing: it measures the marginal impact of a variable on the average treatment effect, however, it is computationally infeasible for even moderate dimensions. Finally, we discuss the widely used method among pharmacoepidemiologists,  hdPS~\citep{schneeweiss2009high}. The hdPS is a multi-step algorithm, we will focus on three crucial steps: 1. ranking variables by proportion of non-zero values and selecting the top $J$, 2. converting each variable into multiple (upto three) binary variables and, 3. ranking variables by the Bross formula and selecting the top $K$. The first step above can lead to excluding important variables and should be omitted~\citep{schuster2015role}. Step two is needed to use the existing Bross formula which requires all variables to be binary. However, this discretization means we can no longer rank the original variables, i.e., hdPS can be used as a ranking system only when all variables are binary.
%interpreting the ranking of these binary variables is difficult and does not translate into variable importance of the original variables. 
%Thus, hdPS can only rank  binary variables. %No one proposal addresses all our requirements for a confounder ranking system for high-dimensional data. 

%The quantity of interest is an extension of the Bross formula %~\citep{bross1966spurious} or the apparent relative risk~\citep{walker1991observation}. 
The intuition behind the Bross formula, and our extension of it, is as follows: if no relationship exists between the exposure and outcome, then any apparent measure of association is due to a third, confounding variable. Thus, the relative risk under the assumption of no exposure-outcome relationship is a measure of confounding effect. To be precise, consider random variables $(O,E,C)\sim P$ for some measure $P$ where $E \in \{0,1\}$ is the exposure, $O\in \mathbb{R}$ the outcome and $C\in \mathbb{R}$ is an additional covariate. We assume the \emph{spurious effect hypothesis}~\citep{bross1966spurious}: $O\perp E|C$ and expectations under this hypothesis are denoted by a $\dagger$. The confounding impact can then be measured by: 
%the relative risk or risk difference under the spurious effect hypothesis:
%The exposure-outcome effect size under the spurious effect hypothesis will then be a measure of confounding impact, we consider the two parameters:
\begin{equation}
	\label{eqn:poiUH}
	\psi(P) = \frac{\mathbb{E}_P^{\dagger}(O|E=1)}{\expec_P^{\dagger}(O|E=0)}, \quad \phi(P) = \expec_P^{\dagger}(O|E=1) - \expec_P^{\dagger}(O|E=0).
\end{equation}
To derive estimable quantities from \eqref{eqn:poiUH}, we first note that the spurious effect hypothesis implies
%\begin{equation}
%\label{eqn:spuriousEffect}
$\expec_P^{\dagger}(O|E=1, C = c) = \expec_P^{\dagger}(O|E=0, C = c) = \expec_P^{\dagger}(O|C = c) = \expec_P(O|C = c), 
$
%\end{equation} 
where the last equality is because  we made no assumptions about the distribution of $O|C$. We then simplify (\ref{eqn:poiUH}): 
%\begin{align*}
$
\expec_P^{\dagger}(O|E=e) = \frac{\expec_P^{\dagger}(1_{E=e}O)}{\pr(E=e)}= \frac{\expec_P^{\dagger}\left[ 1_{E=e}\expec_P^{\dagger}(O|E,C) \right] }{\pr(E=e)} = \frac{\expec_P^{\dagger}\left[ 1_{E=e}\expec_P(O|C) \right] }{\pr(E=e)} 
= \expec_P\left\{\expec_P(O|C)| E=e \right\}
. 
$
%\end{align*}
Denoting $\tau(c) = \expec_P(O|C=c)$, our parameters of interest are
\begin{equation}
	\label{eqn:poi1}
	\psi(P) = \frac{\expec_P\{\tau(C)|E=1\} }{\expec_P\{\tau(C)|E=0\}}, \quad \phi(P) =  {\expec_P\{\tau(C)|E=1\} }- {\expec_P\{\tau(C)|E=0\}}.
\end{equation}
For brevity, we will drop the subscript $P$ when it is clear from context. Motivation for $\psi$ and $\phi$ becomes clearer by considering their values when $C$ is not a confounder. If $C$ is independent of either the exposure $E$ or the outcome $O$, then we immediately see that $\psi(P) = 1$ and $\phi(P) = 0$. Values of $\psi$ and $\phi$ away from 1 and 0, respectively, indicate presence of a confounding effect.

\subsection{High-dimensional Confounder Ranking and Selection}
\label{sec:poi}
Exploiting the intuition of (\ref{eqn:poi1}), we apply them to each variable for a set of potential confounders. Consider i.i.d. data $(O_i,E_i,C_{1i},\ldots, C_{pi}) \sim P$  $(i=1,\ldots, n)$.
% for some probability measure $P$. 
Our target parameters are the \emph{ratio scores}, $\psi_j(P)$, and \emph{difference scores}, $\phi_j(P)$: 
\begin{equation}
	\psi_j(P) = \frac{\expec\{\tau_j(C_{ji})|E_i=1\} }{\expec\{\tau_j(C_{ji})|E_i=0\}},\qquad \phi_j(P) =  {\expec\{\tau_j(C_{ji})|E_{i}=1\} }- {\expec\{\tau_j(C_{ji})|E_i=0\}}, \label{eqn:psiphiD1}
\end{equation}
%\begin{equation}
%	 \phi_j(P) =  {\expec\{\tau_j(C_{ji})|E_{i}=1\} }- {\expec\{\tau_j(C_{ji})|E_i=0\}} \label{eqn:phiD2},
%\end{equation}
where $\tau_j(c) = \expec(O_i|C_{ji} = c)$. Thus we can rank the $p$ covariates according to either $|\psi_j(P) - 1|$ or $|\phi_j(P)|$. We can also define the parameters $\psi$ and $\phi$ in terms of the propensity score, $\pi_j(c) = \pr(E_i=1|C_{ji} = c)$. This is demonstrated in the following lemma~(proof in Web Appendix A). 
\begin{lemma}
	\label{lemma:propScore}
	For a random vector $(O_i, E_i, C_{1i},\ldots, C_{pi})\sim P$ the parameters (\ref{eqn:psiphiD1}) can equivalently be written as $\psi_j(P) = {[\expec\{O_i\pi_j(C_{ji})\}/\pr(E_i=1) ]}/\{\expec[O_i\{1 - \pi_j(C_{ji})\}]/\pr(E_i=0) \}$, and $\phi_j(P) =  {\expec\{O_i\pi_j(C_{ji})\}/\pr(E_i=1) } - {\expec[O_i\{1 - \pi_j(C_{ji})\}]/\pr(E_i=0)}$.	
\end{lemma}
Note that our measures of confounder importance are in terms of conditional expectations, which can be calculated for any variable type (given existence of first moments). In contrast, the Bross formula of the hdPS method is in terms of relative risk and prevalence of $C_i$, which can only be calculated for binary $C_i$. Furthermore, a scoring system for grouped confounder ranking follows from a straightforward extension of our parameters $\psi_j$ and $\phi_j$. Consider i.i.d data $(O_i, E_i, \bs{C}_{1i},\ldots, \bs{C}_{pi})\sim P$ where $\bs{C}_{ji}\in \mathbb{R}^{p_j}$ is now a group of covariates. The grouped ratio and difference scores are defined as
%\begin{align}
$\Psi_j(P) = {\expec\{\bs{\tau}_j(\bs{C}_{ji})|E_i=1\} }/{\expec\{\bs{\tau}_j(\bs{C}_{ji})|E_i=0\}}$, %\label{eqn:psiGD1}\\
$\Phi_j(P) =  {\expec\{\bs{\tau}_j(\bs{C}_{ji})|E_{i}=1\} }- {\expec\{\bs{\tau}_j(\bs{C}_{ji})|E_i=0\}},$
%\label{eqn:phiGD1},
%\end{align}
where $\bs{\tau}_j(\bs{c}) = \expec(O_i|\bs{C}_{ji} = \bs{c})$. As in Lemma~\ref{lemma:propScore}, we can also define the parameters $\Psi_j$ and $\Phi_j$ in terms of the propensity score, $\bs{\pi}_j(\bs{c}) = \pr(E_i=1|\bs{C}_{ji} = \bs{c})$.

Say we have estimates $\wh{\psi}_j$ or $\wh{\phi}_j$ for each variable in our data, we can then rank variables by confounder importance and select the top $K$. Alternatively, if we can show our estimates to be asymptotically normal we can conduct an $\alpha$-level test for $H_0: \psi_j = 1$ vs $H_1: \psi_j \not= 1$ (or $H_0: \phi_j = 0$ { vs } $H_1: \phi_j \not= 0)$. Thus the user needs to only specify the cut-off $K$ or level $\alpha$. Asymptotically normal and efficient estimators for $\psi_j, \phi_j$ are presented in a later Section.

\subsection{More on the Bross Formula}
\label{sec:brossFormula}

Before presenting our estimation framework, we address the main criticisms of using the Bross formula for confounder selection and, we also outline some of its limitations. 

The primary issue is that the Bross formula measures the confounding impact of a single variable and does not naturally  extend to multiple variables. The original Bross formula was designed under a single potential confounder, sequentially applying this multiple variables ignores the \emph{joint confounding effect}. This seems akin to fitting multiple univariate regression models as opposed to a single multivariate model. Another well known issue of the Bross formula is how it handles instrumental variables: variables associated with the exposure only will not have a Bross score of zero. We now address these issues, making a case for the utility of the Bross formula and consequently our proposal.

The hdPS proposal (a special case of our framework) has been shown to \emph{work} in a number of simulation settings, including plasmode studies where investigators aim to replicate real data~\citep{patrick2011implications, franklin2014plasmode,franklin2015regularized,franklin2017comparing,pang2016targeted,pang2016effect,guertin2016head,schneeweiss2017variable,karim2018can}. To understand this behavior, we explore the theoretical quantity estimated by our parameters  $\psi_j$ and $\phi_j$. Consider the following simple data generating mechanism:
$\expec(O|E,C_1,\ldots,C_p) = \beta_0 + \theta E +\sum_{j=1}^{p} \beta_j C_j$,  
$\pr(X=1|C_1,\ldots,C_p) = \sum_{j=1}^{p}\alpha_jC_j\, : \sum_{j=1}^{p}\alpha_j = 1$  and  $\alpha_j> 0$ ,
$C_j \overset{i.i.d.}{\sim} Unif[0,1]$.
By straightforward calculation we obtain
%\begin{align*}
%\label{eqn:tauSimple}
%\tau_j(C_j) %&= \beta_0 + \beta_jC_j + \theta\expec(X|C_j) + \sum_{k\not= j} \beta_k\expec(C_k)\\
%&= \beta_0 + \beta_jC_j + \theta\Big\{ \alpha_j C_j + \sum_{k\not= j} \alpha_k \expec(C_k) \Big\} + \sum_{k\not= j} \beta_k\expec(C_k)\\
$\tau_j(C_j) = \beta_0 + \beta_jC_j + \theta\Big\{ \alpha_j C_j + \sum_{k\not= j} \alpha_k/2 \Big\} + \sum_{k\not= j} \beta_k/2,$
%\end{align*}
and thus for $j=1,\ldots, p$ our parameter is
%\begin{align*}
%	\label{eqn:phijSimple}
%\phi_j %&= \{\beta_j + \theta \alpha_j \} \{\expec(C_j|X=1) - \expec(C_j|X=0) \}\\
%&= \{\beta_j + \theta \alpha_j \} \Big[ \frac{2\alpha_j}{3} + \sum_{k\not= j} \alpha_k\expec(C_k) - \Big\{1 - \frac{2\alpha_j}{3} + \sum_{k\not= j} \alpha_k\expec(C_k) \Big\} \Big]\\
%&= \{\beta_j + \theta \alpha_j \} \Big\{ \frac{4\alpha_j}{3} + \sum_{k\not= j} \alpha_k - 1 \Big\}\\
%&= \{\beta_j + \theta \alpha_j \} \Big\{ \frac{4\alpha_j}{3} + (1-\alpha_j) - 1 \Big\}\\
$\phi_j = ({\alpha_j}/{3})(\beta_j + \theta \alpha_j )$.
%\end{align*}
We see that $\phi_j$ in this case depends on the combined association between $C_j$ and both the exposure \emph{ and} outcome. Variables with a large Bross score will have large $\alpha_j$ and $\beta_j$ values. This suggest a valid ranking system for confounder importance; recall that we are only interested in confounder ranking, not the magnitude or range of scores. The example also illustrates what happens with instrumental variables $(\beta_j = 0)$: while the score will not be zero, if the confounder-exposure effect sizes are relatively small compared to the confounder-outcome effect sizes, instrumental variables will have a lower score compared to true confounders.

The above example illustrates how our proposed confounder scoring framework can be valid even if we ignore the joint confounder effect. Based on existing work on simulation studies of hdPS, we expect to see similar results in more complex data generating mechanisms. Finding a minimum set of assumptions under which our proposed framework is valid is an interesting open problem. 

\section{Our Estimation Framework}
\label{sec:estimation}

\subsection{Naive Plug-in Estimation}
Let $\wh{\tau}_j$ and $\wh{\pi}_j$ denote some estimates for $\tau_j$ and $\pi_j$, respectively. Denote by $\wh{P}^{om}$ a probability measure corresponding to our estimate $\wh{\tau}_j$, then plug-in estimators for the scores are
\begin{align*}
	%\label{eqn:pluginOut}
	\psi_j(\wh{P}^{om}) &= \frac{\expec_{P_n}\{\wh{\tau}_j(C_{ji})|E_i=1\} }{\expec_{P_n}\{\wh{\tau}_j(C_{ji})|E_i=0\}}, \\
	\phi_j(\wh{P}^{om}) &= {\expec_{P_n}\{\wh{\tau}_j(C_{ji})|E_{i}=1\} }- {\expec_{P_n}\{\wh{\tau}_j(C_{ji})|E_i=0\}},
\end{align*}
where $P_n$ is the empirical distribution of our data. Similarly, let $\wh{P}^{ps}$ be the measure associated with $\wh{\pi}_j$, then alternative plug-in estimators are
\begin{align*}
	%	\label{eqn:pluginPS}
	\psi_j(\wh{P}^{ps}) &= \frac{\expec_{P_n}\{O_i\wh{\pi}_j(C_{ji})\}/\expec_{P_n}(E_i) }{\expec_{P_n} [O_i\{1 - \wh{\pi}_j(C_{ji})\}]/\expec_{P_n}(1 - E_i)},\\
	\phi_j(\wh{P}^{ps}) &= {\expec_{P_n}\{O_i\wh{\pi}_j(C_{ji})\}}/{\expec_{P_n}(E_i)} - {\expec_{P_n}[O_i\{1 - \wh{\pi}_j(C_{ji})\}]}/{\expec_{P_n}(1-E_i)}.
\end{align*}

These na\"{i}ve plug-in estimators face a number of limitations. Firstly, they can suffer from model misspecification. 
%as they depend on estimated models for the outcome model or propensity score. 
Secondly, any asymptotic results might only be tractable for limited parametric model choices for $\wh{\tau}_j$ and $\wh{\pi}_j$. Thirdly, there is no guarantee of statistical efficiency, i.e., there might be other consistent estimators with lower variance.
%A general purpose solution would be to use bootstrap estimates for generating confidence intervals and $p$-values. 
In the next section, we present an estimation framework which leads to efficient, asymptotically normal and doubly robust estimators. The double robustness property means that we need only correctly estimate either $\pi_j$ or $\tau_j$ (without knowledge of which model is correctly specified). Furthermore, our estimators can use advanced machine learning methods for estimating $\pi_j$ and $\tau_j$.  

\subsection{Efficient Estimation}

In this section, we will present two estimators: the first performs a one-step correction to the na\"{i}ve plug-in estimator, and the second is a targeted maximum likelihood estimator~(TMLE) following the framework of \cite{laan2011targeted}. We begin with a reparametrization of our parameters $\psi_j$ and $\phi_j$ as follows: 
\begin{equation}
	\label{eqn:psiTheta}
	\psi_j(P) %\equiv \psi_j(\theta_j, \mu_E, \mu_O;P) 
	= \frac{\theta_j(P)/\mu_E(P)}{\{ \mu_O(P) - \theta_j(P) \}/\{ 1 - \mu_E(P) \}},
\end{equation}
\begin{equation}
	\label{eqn:phiTheta}
	\phi_j(P) %\equiv \phi_j(\theta_j, \mu_E, \mu_O;P) 
	= {\theta_j(P)/\mu_E(P)} - {\{ \mu_O(P) - \theta_j(P) \}/\{ 1 - \mu_E(P) \}},
\end{equation}
where $\theta_j(P) = \expec_P\{I(E_i=1) \tau_j(C_{ji}; P ) \}, \mu_O(P) = \expec_P(O_i)$ and $\mu_E(P) = \expec_P(E_i)$. By using the sample means as estimates for $\mu_O$ and $\mu_E$, and the delta method, it is sufficient to establish asymptotic normality for $\theta_j$. In the remainder of this section, we will first present estimators for $\theta_j$ and estimators for (\ref{eqn:psiTheta}) and (\ref{eqn:phiTheta}) will follow.

%
%For i.i.d. data $\bs{X}_i = (O_i, E_i, C_{1i}, \ldots, C_{pi})\sim P \ (i=1,\ldots,n)$  for some probability measure $P$, define 
%\begin{equation}
%	\label{eqn:simpleParas}
%	\theta_j(P) = \expec_P\{I(E_i=1) \tau_j(C_{ji}; P ) \}, \quad \mu_O(P) = \expec_P(O_i), \quad \mu_E(P) = \expec_P(E_i),
%\end{equation}
%where $\tau_j(c; P) = \expec_P(O_i|C_{ji} = c)$. Then, our variable importance parameters can be written as: %\eqref{eqn:simpleParas}: 
%
%With this reparametrization, we observe that it is sufficient to find an asymptotically efficient estimator for $\theta_{j}(P)$. 

%
%As before, we denote denote the propensity score by $\pi_j(c; P) = \pr(E=1|C_j=c)$. We define the \emph{adjusted exposure-response model} as $Q_{j}(e, c_j; P) = \expec_P(O|E=e, C_j = c_j)$. The functions $\pi_j, \tau_j$ and $Q_j$ are related by the following identity:
%\begin{equation}
%	\label{eqn:3func}
%	\tau_j(c_j;P) = \pi_j(c_j;P)Q_j(1, c_j;P) + \{1-\pi_j(c_j;P)\}Q_j(0,c_j;P).
%\end{equation}
%
%For any $P$-measurable function $f$, we define $Pf = \int f(\bs{x})\, dP(\bs{x})$ and $P_nf = n^{-1}\sum_{i=1}^{n}f(\bs{x}_i)$ for an i.i.d. sample $\bs{x}_i\sim P\, (i=1,\ldots,n)$. More generally, for a measure $P$ we denote its empirical probability measure by $P_n$. We define a \emph{statistical model} as a collection of probability measures denoted by $\mathcal{M}$. Throughout this paper, we consider nonparametric statistical models. 

\subsubsection{One-step correction}
\label{sec:oneStep}
Denote by $\wh{\tau}_{j,n}$ and $\wh{\pi}_{j,n}$, some estimators for $\tau_j$ and $\pi_j$, respectively. We propose a one-step correction to the naive plug-in estimator $\wh{\theta}^{naive}_{j,n} = n^{-1}\sum_{i=1}^{n} I(E_i=1)\wh{\tau}_{j,n}(C_{ji})$:
\begin{align}
	\wh{\theta}^{dr}_{j,n} %&= \theta_{j}(\wh{\pi}_{j,n}^{dr},\wh{\tau}_{j,n}^{dr}) + n^{-1}\sum_{i=1}^{n} D_{\theta_{j}(\wh{\pi}^{dr}_{j,n}, \wh{\tau}_{j,n}^{dr})} \nonumber\\
	&=  \wh{\theta}^{naive}_{j,n} +  \frac{1}{n}\sum_{i=1}^{n}  O_i \wh{\pi}_{j,n}(C_{ji}) - \wh{\tau}_{j,n}(C_{ji}) \wh{\pi}_{j,n}(C_{ji}). \label{eqn:DRestimate}
\end{align}
This corrected estimator is doubly robust: $\wh{\theta}^{dr}_{j,n} \to \theta_j(P^0)$ as long as one of our estimators $\wh{\tau}_{j,n}$ and $\wh{\pi}_{j,n}$ is consistent. We formalize and prove this result in Web Appendix A.

Furthermore,  the simple one step correction makes it easy to implement and computationally efficient. However, its main limitation is that the one step correction can lead to an estimate outside the parameter space. For example, for outcome $O_i\in [0,1]$ we must have $\phi_j\in [-1,1]$, but our estimate $\wh{\phi}^{dr}_{j,n}$ may be outside this range. Our second proposed estimator is a plug-in estimator which overcomes this issue whilst maintaining the desirable properties of $\wh{\theta}_{j,n}^{dr}$.

\subsubsection{Targeted Maximum Likelihood Estimator}
\label{sec:tmle}

Targeted maximum likelihood estimation~(TMLE), is a general framework for obtaining efficient plug-in estimators. A high level summary of this procedure is as follows: we begin with some initial estimate $\wh{P}_0$ of $P^0$~(the measure of our observed data), we then iteratively update $\wh{P}_0$, generating a sequence of estimates $\{\wh{P}_k\}_{k=0}^{K}$ until some convergence criteria is met, finally we obtain the plug-in estimator $\psi_j(\wh{P}_K)$ (or $\phi_j(\wh{P}_K)$). For a more detailed discussion of the general TMLE framework, see \cite{laan2011targeted}.

We now present the algorithm for obtaining the targeted maximum likelihood estimator~(TMLE) for $\theta_j(P^0)$, details of the derivation of the TMLE are relegated to Web Appendix B. For the algorithm, we require the \emph{adjusted exposure-response model} defined as $Q_{j}(e, c_j; P) = \expec_P(O|E=e, C_j = c_j)$. The use of $Q_j$ is a technical requirement which we detail in Web Appendix B, briefly, our TMLE algorithm exploits the decomposition of the joint density~(or Radon-Nikodym derivative): $p(o,e,c_j) = p(o|e,c_j)p(e|c_j)p(c_j)$. 
Note that the functions $\pi_j, \tau_j$ and $Q_j$ are related by the identity, $\tau_j(c_j;P) = \pi_j(c_j;P)Q_j(1, c_j;P) + \{1-\pi_j(c_j;P)\}Q_j(0,c_j;P)$. The TMLE, denoted by $\wh{\theta}_j^{tl}$ is obtained by the following algorithm:

\begin{enumerate}
	\item Initialize estimators $\wh{\pi}^0_j, \wh{Q}^0_j$ for $\pi_j(\cdot;P^0)$ and $Q_j(\cdot, \cdot; P^0)$; and $\epsilon_1^0 = \epsilon_2^0 = 1$.
	\item For $k=0,1,\ldots$ until $\epsilon_1^k = \epsilon_2^k = 0$.
	\begin{enumerate}
		\item Update the propensity score $\wh{\pi}^{k+1}_j \gets \wh{\pi}^{k}_j(\epsilon_1^k)$ where $\wh{\pi}^{k}_j(\epsilon) = \expit\{\logit(\wh{\pi}_j^{k}) + \epsilon H_1^{k}\}$, $H_1^k(\cdot) = -2\wh{\pi}^{k}_j(\cdot)\{\wh{Q}_j^{k}(1,\cdot) - \wh{Q}_j^{k}(0, \cdot)\} - \wh{Q}_j^{k}(0, \cdot)$ and,
		\begin{align*}
			%H_1^k(\cdot) &= -2\wh{\pi}^{k}_j(\cdot)\{\wh{Q}_j^{k}(1,\cdot) - \wh{Q}_j^{k}(0, \cdot)\} - \wh{Q}_j^{k}(0, \cdot),\label{eqn:H1}\\
			\epsilon_1^k &= \underset{\epsilon \in \mathbb{R}}{\operatorname{argmax} }\  \frac{1}{n} \sum_{i=1}^{n} E_i\log\{\wh{\pi}^{k}_j(\epsilon)\} + (1-E_i)\log\{ 1 - \wh{\pi}^{k}_j(\epsilon) \}.
		\end{align*}
		\item Update the adjusted exposure-response model $\wh{Q}_j^{k+1}\gets \wh{Q}_j^k(\epsilon_2^k)$ where $\wh{Q}_j^k(\epsilon) = \wh{Q}_j^{k} + \epsilon H^{k}_2$, $H_2^k(\cdot) = -\wh{\pi}^{k+1}_j(\cdot)$ and,
		\begin{align*}
			%H_2^k(\cdot) &= -\wh{\pi}^{k+1}_j(\cdot),\\
			\epsilon_2^k &= \underset{\epsilon \in \mathbb{R}}{\operatorname{argmin} }\  \frac{1}{2n} \sum_{i=1}^{n} \{O_i - \wh{Q}_j^k(\epsilon)\}^2.
		\end{align*}
	\end{enumerate} 
	\item Return the parameter estimate
	\begin{equation}
		\label{eqn:tmleEst}
		\wh{\theta}_{j,n}^{tl} = \frac{1}{n}\sum_{i=1}^{n} \{I(E_i=1)\wh{\tau}_j^{k+1}(C_{ji})\},
	\end{equation}
	where 	$\wh{\tau}^{k+1}_j(c) = \wh{\pi}^{k+1}_j(c)\wh{Q}^{k+1}_j(1, c) + \{1-\wh{\pi}^{k+1}_j(c)\}\wh{Q}^{k+1}_j(0,c)$.
\end{enumerate}

\begin{note}
	The TMLE is a \emph{substitution estimator}, which means that as long as the range of $\wh{\pi}_j^{0}$ and $\wh{Q}_j^0$ is valid, our TMLE will also be within the parameter space.   
\end{note}
\begin{note}
	The update in Step 2(b) is designed for a continuous outcome variable. By appropriately defining the path $\wh{Q}_j^k(\epsilon)$ and loss function, we can build a TMLE for other types of outcome variables. We detail this in Web Appendix B.
\end{note}
\begin{note}
	Both estimators, TMLE and doubly robust, rely on initial estimators for $\pi_j(\cdot, P^0)$ and  $Q_j(\cdot, \cdot; P^0)$. While a semi-parametric or nonparametric approach may seem computationally expensive, recall that $\phi_j$ and $\psi_j$ are univariate functions. Thus, we can utilize fast algorithms for nonparametric estimators such as wavelets. Alternatively, we could use polynomial regression which we found to be sufficient for our numerical experiments.
\end{note}
\begin{note}
	The above algorithm can also be used for grouped confounder ranking, i.e. estimating $\Psi_j$ and $\Phi_j$. In this case estimating $Q_j$ and $\pi_j$ can become computationally demanding. However, in most cases of high-dimensional data, either each variable group will be small, or for large sized groups the total number of groups will be small. Furthermore, in many data applications~(including our analysis of the Boston Housing data), simple parametric models can serve as sufficient estimators. 
\end{note}

\subsubsection{Asymptotic Normality of Estimators}
\label{sec:theory}

In this section, we establish asymptotic normality of our estimators. Theorem~\ref{thm:one} establishes asymptotic normality for $\wh{\theta}^tl_{j,n}$ and $\wh{\theta}_{j,n}^{dr}$. This leads to Corollary~\ref{corr:main}, which establishes asymptotic normality for estimates of $\phi_j$ and $\psi_j$. Before presenting our main results, we require some technical background and definitions.  

Firstly, our results crucially rely on the efficient influence curve~(EIC) of a parameter. Roughly, the influence curves of a parameter, $\psi(P)$, are essentially derivatives of the functional $\psi$ with respect to $P$. Alternatively, we can define an influence curve for \emph{asymptotically linear estimators}, an estimator $\wh{\psi}$ is said to be asymptotically linear if it satisfies
\begin{equation}
	\label{eqn:asymLin}
	\wh{\psi} - \psi = \frac{1}{n}\sum_{i=1}^{n} D_{\psi}(\bs{X}_i) + o_p(n^{-1/2}).
\end{equation}
We then define $D_{\psi}$ as an influence curve. For non-parametric modeling of $\pi_j$ and $Q_j$, there is only one influence curve called the EIC. \cite{bickel1993efficient} show that if an estimator is asymptotically linear, and its influence curve is the EIC, then it is asymptotically efficient. Form \eqref{eqn:asymLin}, it follows from the CLT that $\sqrt{n}(\wh{\psi} - \psi)$ is asymptotically normal with variance $Var\{D_{\psi}(\bs{X}_i)\}$. The EIC for $\theta_j$ (derived in Web Appendix B) is
\begin{equation}
	\label{eqn:eicTheta}
	D_{\theta_j(P)}(\bs{x}) \equiv D_{\theta_j(\pi_j, \tau_j)}(\bs{x}) = o\pi_j(c_j; P) + \tau_j(c_j;P)\{I(e=1) - \pi_j(c_j;P)\} - \theta_j(P), 
\end{equation}
for a vector $\bs{x} = (o,e,c_1,\ldots,c_p)$. 

We now present and discuss the main conditions needed for our theoretical results. 

\begin{condition}
	\label{assum:slowConvergence}
	For initial estimators $\wh{\tau}_j^n$ and $\wh{\pi}_j^n$ of $\tau_j(P^0)$ and $\pi_j(P^0)$, respectively, we have  $\|\wh{\tau}^n_j(c) - \tau_j(c;P^0)\|_{\infty} = o_p(n^{-1/4})$ and $\|\wh{\pi}^n_j(c) - \pi_j(c;P^0)\|_{\infty} = o_p(n^{-1/4})$.
\end{condition}
\begin{condition}
	\label{assum:Donsker}
	With initial estimators $\wh{\tau}_j^n$ and $\wh{\pi}_j^n$, the function $\bs{x} \mapsto D_{\theta_j(\wh{\tau}_j^n, \wh{\pi}_j^n)}(\bs{x})$ is $P^0$--Donsker~\citep{vaart1998asymptotic}.
\end{condition}

Condition~\ref{assum:slowConvergence}, implies that we only require $n^{1/4}$--convergence of our estimates of the nuisance parameters $\pi_j$, $\tau_j$, as opposed to usual $n^{1/2}$--convergence. Using advanced machine learning tools, this condition can be easily satisfied. Condition~\ref{assum:Donsker}, is a standard condition in empirical process theory~\citep{vaart1996weak,vaart1998asymptotic}. It depends on the underlying complexity of our nonparametric modeling class for functions $\tau_j$ and $\pi_j$. With the above conditions, we now establish asymptotic linearity of $\wh{\theta}_{j,n}^{tl}$ and $\wh{\theta}_{j,n}^{dr}$ with EIC \eqref{eqn:eicTheta}.

\begin{theorem}
	\label{thm:one}
	For data $\bs{X}_i = (O_i, E_i, C_{1i}, \ldots, C_{pi})\sim P^0$, under Conditions \ref{assum:slowConvergence} and \ref{assum:Donsker},
	\begin{equation*}
		\label{eqn:main}
		\wh{\theta}_{j,n}^{tl}  - \theta_j(P^0) = \frac{1}{n}\sum_{i=1}^{n} D_{\theta_j(P^0)}(\bs{X}_i) + o_p(n^{-1/2}),
	\end{equation*}
	where $D_{\theta_j(P)}$ is the EIC~\eqref{eqn:eicTheta}. A similar result holds with $\wh{\theta}_{j,n}^{tl}$ replaced by $\wh{\theta}_{j,n}^{dr}$.
\end{theorem}

The above theorem establishes asymptotic linearity, which implies asymptotic normality by the CLT and asymptotic efficiency by use of the EIC. By the delta method, we can then obtain a similar result for our parameters $\psi_j$ and $\phi_j$, this is formalized in the following corollary. 

\begin{corollary}
	\label{corr:main}
	Under the conditions of Theorem~\ref{thm:one}, 
	\begin{align*}
		\label{eqn:ExpansionMain}
		\left[\begin{array}{c}
			\wh{\psi}^{tl}_{j,n}\\
			\wh{\phi}^{tl}_{j,n}
		\end{array}\right] - 		\left[\begin{array}{c}
			\psi_j(P^0)\\
			\phi_j(P^0)
		\end{array}\right] = \frac{1}{n} \sum_{i=1}^{n} \left[\begin{array}{c}
			D_{\psi_j(P^0)}(\bs{X}_i)\\
			D_{\phi_j(P^0)}(\bs{X}_i)
		\end{array}\right] + o_p(n^{-1/2}),
	\end{align*}
	where the influence curves $D_{\psi_j(P^0)}$ and $D_{\phi_j(P^0)}$ are defined as 
	\begin{align*}
		D_{\psi_j(P^0)} %&=  &&\frac{\mu_O(P^0)\{1-\mu_E(P^0)\} D_{\theta_{j}(P^0)}}{\mu_E(P^0)\{\mu_O(P^0) - \theta_{j}(P^0)\}^2}  \\
		%& \ &&- \frac{\theta_j(P^0)\{1- \mu_E(P^0)\}D_{\mu_O(P^0)}}{\mu_E(P^0)\{\mu_O(P^0) - \theta_{j}(P^0)\}^2} - %\frac{\theta_j(P^0)D_{\mu_E(P^0)}}{\{\mu_E(P^0)\}^2\{\mu_O(P^0) - \theta_j(P^0)\} }, \\
		%&=  &&\frac{\mu_O(P^0)\psi_j(P^0) D_{\theta_{j}(P^0)}}{\theta_j(P^0)\{\mu_O(P^0) - \theta_{j}(P^0)\}} - \frac{ \psi_j(P^0)D_{\mu_O(P^0)} %}{\{\mu_O(P^0) - \theta_{j}(P^0)\}} - \frac{\{1-\mu_E(P^0)\} \psi_j(P^0)D_{\mu_E(P^0)}}{\{\mu_E(P^0)\} }.\\
		&= &&\frac{\psi_j(P^0)}{\mu_O(P^0) - \theta_{j}(P^0)} \left\{ \frac{\mu_O(P^0)D_{\theta_j(P^0)}}{\theta_j(P^0)} - {D_{\mu_O(P^0)}} \right\} - \left\{\mu^{-1}_E(P^0) - 1\right\}\psi_j(P^0)D_{\mu_E(P^0)},
	\end{align*}
	\begin{align*}
		D_{\phi_j(P^0)} & = \frac{D_{\theta_j(P^0)}}{\mu_E(P^0)\{1 - \mu_E(P^0)\}} - \frac{D_{\mu_O(P^0)}}{1 - \mu_E(P^0)} - \frac{\{\mu_{O}(P^0) - \theta_j(P^0)\}D_{\mu_E(P^0)}}{\{1  - \mu_E(P^0)\}^2} - \frac{\theta_j(P^0)D_{\mu_E(P^0)}}{\mu_E(P^0)},
	\end{align*}
	where $\mu_X(P) = \mathbb{E}_PX$ is the expectation of $X\sim P$, and $D_{\mu_X(P)}(x) = x - \mu_X(P)$ is its EIC. 	
\end{corollary}

\section{Numerical Experiments}
\label{sec:simulations}

In this section, we  assess the empirical performance of our estimators via simulations. We implement our doubly robust estimator~\eqref{eqn:DRestimate} and targeted maximum likelihood estimator~\eqref{eqn:tmleEst}, henceforth referred to as hdCS-1 and hdCS-2, respectively. We divide our numerical experiments into three sections, namely experiments for low-dimensional data, high-dimensional data and misspecified modeling. In the low-dimensional setting, we compare our proposed methods to  C-TMLE~\citep{laan2010collaborative}, SC-TMLE~\citep{ju2019scalable} and outcome adaptive lasso (OAL)~\citep{shortreed2017outcome}. For the high-dimensional setting we exclude SC-TMLE and C-TMLE. The reason for exclusion is high-computational time as C-TMLE and SC-TMLE do not scale well with increasing dimension $p$; even in the low-dimensional case we observed SC-TMLE to be 10-25 times slower than hdCS or OAL. Finally, we note that OAL relies heavily on correct specification of the outcome and propensity score model, in the final section we consider deviations from these modeling assumptions.

\subsection{Low-Dimensional Data}
\label{sec:lowDimSim}

For this simulation study, we generated data with sample size, $n=500$ and dimension $p = 30$. The set of potential confounders are generated as $\bs{C}_i = (C_{i1},\ldots, C_{ip}) \sim \mathcal{N}_p(\bs{0}, \Sigma)$ where $\Sigma_{jk} = \rho^{|j-k|}$ with $\rho = 0.3$ or $0$ (with the convention $0^0 = 1$). We generate a binary exposure $E_i$, from a Bernoulli distribution with $\mathrm{logit}\{Pr(E_i = 1|\bs{C}_i) \}  = \sum_{j=1}^{p} \alpha_j C_{ij}$ and, a continuous response $Y$, such that $Y_i = \theta E_i + \sum_{j=1}^{p} \beta_j C_{ij} + \e_i$ where $\e_i\sim \mathcal{N}(0,1)$. We consider two choices for the average treatment effect $\theta = 0$ or 2. From the 30 potential confounders, the first 5 are true confounders (i.e., $\beta_j = 0.6$ and $\alpha_j = 1$ for $j=1,\ldots,5$), the next 5 are precision variables (i.e., $\beta_j = 0.6$ and $\alpha_j = 0$ for $j=6,\ldots, 10$), and the next 5 are instrumental variables (i.e., $\beta_j = 0$ and $\alpha_j = 1$ for $j=11,\ldots, 15$). The remaining 15 variables were spurious.

For brevity, we implement hdCS-1 and hdCS-2 to estimate the difference scores, $\phi_j$, only. We also generated 90\% confidence intervals and selected a variable if the confidence interval contained 0. To speed up computation, for C-TMLE, SC-TMLE and hdCS all initial estimates (e.g. of the outcome model or propensity score) were obtained using (generalized) linear models. All competing methods were implemented with their default setting in \texttt{R}~\citep{rct2014r}. Additionally, for implementing SC-TMLE we use the partial correlation pre-ordering strategy as proposed by \cite{ju2019scalable}.

In Figure~\ref{fig:simLow}, we present the sensitivity and specificity of each method to assess the confounder selection accuracy. We observe little difference between the estimates obtained by hdCS-1 and hdCS-2, both methods exhibit the best performance compared to competing methods. The C-TMLE and SC-TMLE proposals have very poor confounder selection performance. The OAL performs well but has a slightly lower specificity; this is because OAL aims to select true confounders and precision variables in order to obtain estimates for the average treatment effect. In the bottom panel of Figure~\ref{fig:simLow}, we show the average values of our estimates of ${\phi}_j$. When $\theta = 0$ and $\rho = 0$, we clearly see all confounders having a score, $\wh{\phi}_j\approx 0.25$, and all other variables have an estimated score of $0$. When $\theta = 3$, we observe one of the drawbacks of the bross formula, instrumental variables do not have a non-zero value for $\phi_j$. Finally, introducing correlation among the covariates leads to variability in values of $\wh{\phi}_j$. 
%The median run time for hdCS-1, hdCS-2, C-TMLE, SC-TMLE, and OAL were 0.3, 0.6, 30, 8, and 0.8 seconds, respectively.

\begin{figure}
	\centering
	\includegraphics[width=\linewidth]{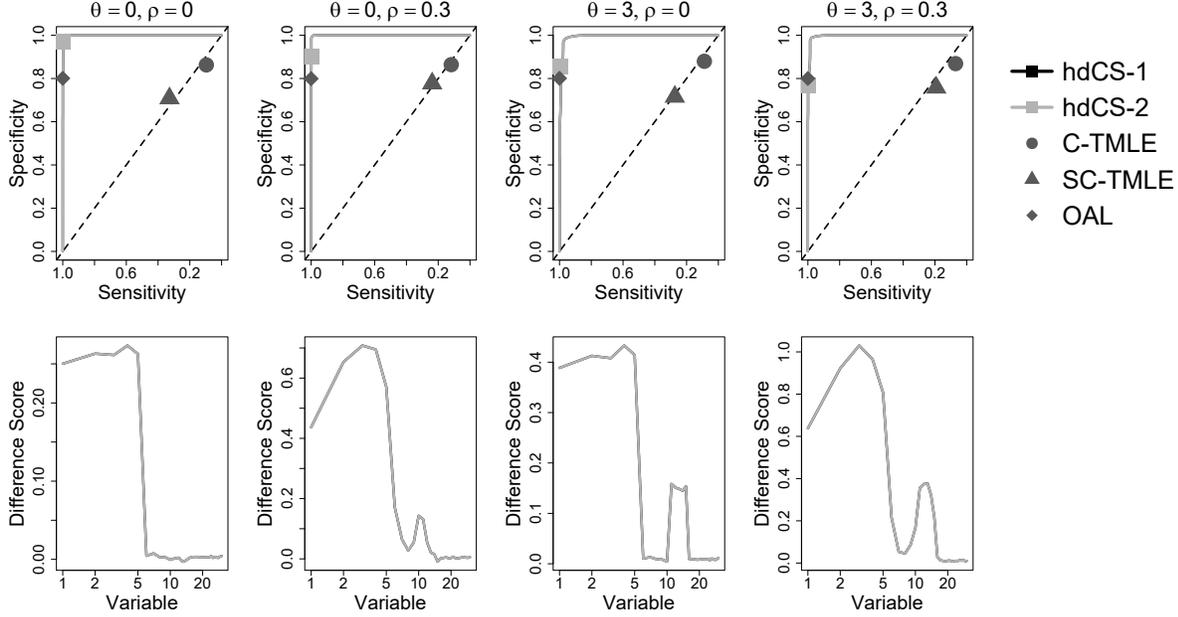}
	\caption{Results of numerical experiments for low dimensional data. \emph{Top}: ROC curves for hdCS-1 and hdCS-2, obtained by varying $k$: the number of variables selected based on difference score. We also plot the sensitivity and specificity of the selected set based for hdCS-1, hdCS-2, C-TMLE, SC-TMLE, and OAL. \emph{Bottom}: Average estimated difference scores, for hdCS-1 and hdCS-2.} 
	\label{fig:simLow}
\end{figure}

\subsection{High-Dimensional Data}
\label{sec:highDimSim}

For this simulation, we generated data with sample size, $n=500$ and dimension $p = 1000$. Data is generated as in the previous section, only now we have 985 spurious variables. In Figure~\ref{fig:simHigh}, we observe very similar performance of both methods as in the low-dimensional case. With a high number of spurious variables, OAL exhibited a higher specificity compared to the low-dimensional case. 

\begin{figure}
	\centering
	\includegraphics[width=\linewidth]{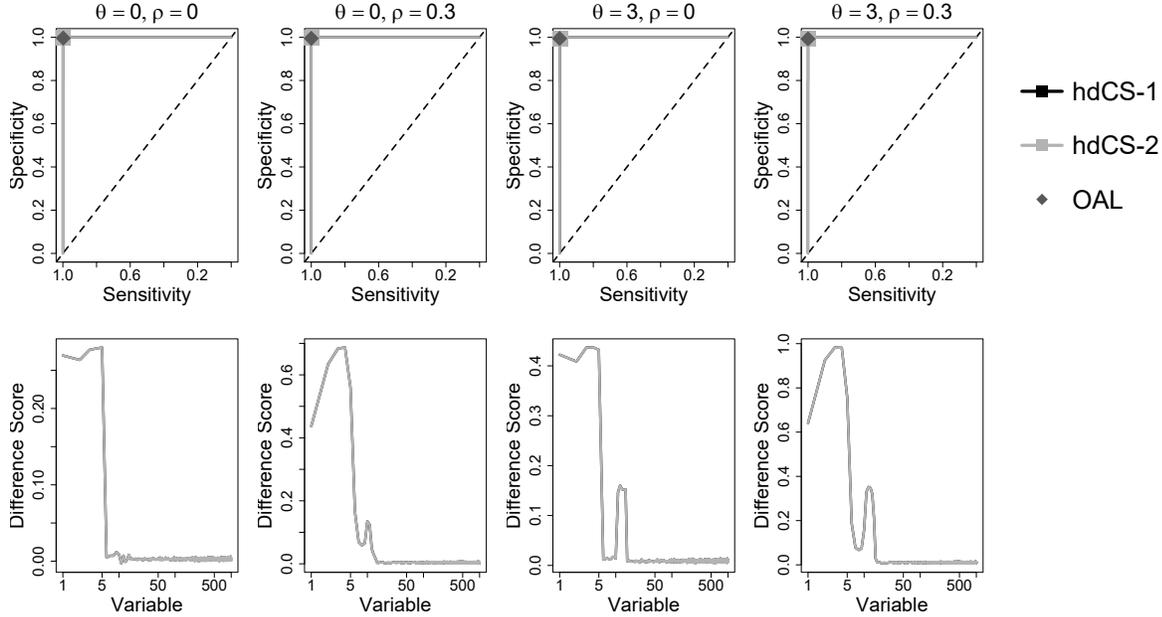}
	\caption{Results of simulation study for high-dimensional data (see label of Figure~\ref{fig:simLow}).}
	\label{fig:simHigh}
\end{figure} 

\subsection{Misspecified Modeling}
\label{sec:simModMis}

For this simulation study, we generated data with $n=500$ and $p=30$ and $50$. The set of potential confounders were generated as $\bs{C}_i\sim \mathcal{N}_p(\bs{0},I)$. We generate a binary exposure $E_i$ and continuous response $Y_i$ where
\begin{align*}
	\mathrm{logit}\{ \Pr (E_i=1|\bs{C}_i)\}&= 	-15 + \sum_{j=1}^{5}3\sin(3C_{ij}) + \sum_{j=11}^{15} (C_{ij}^3 - C_{ij} + 3),\\
	Y_i &= \theta E_i + \sum_{j=1}^{5}1.8\sin(3C_{ij}) + \sum_{j=6}^{10}1.8\cos(4C_{ij}) + \e_i,
\end{align*}
where $\e_i\sim \mathcal{N}(0,1)$ and $\theta\in \{0,3\}$. Allowing for modeling flexibility, we implemented hdCS using polynomial regression (with degree 6) for initial estimates of $Q_j$ and $\pi_j$. We implemented C-TMLE and SC-TMLE via a super-learner using generalized linear and generalized additive models as candidates. As OAL is unable to account for modeling non-linear functions, it was implemented using its default \texttt{R} implementation.

We present the confounder selection properties for each method in terms of sensitivity and specificity in Figure~\ref{fig:simModMis}. In this case, we observe the clear disadvantage of using OAL under a misspecified model. Both hdCS methods seem to do very well exhibiting a high specificity and sensitivity. Looking at the average estimated values for $\phi_j$ in Figure~\ref{fig:simModMis}, we observe some bias in the estimates for non-confounders, particularly spurious variables. This suggests using a more flexible modeling approach such as nonparametric regression, however, as evident from the average sensitivity and specificity values, 95\% confidence intervals for $\phi_j$ contain 0 frequently. 
\begin{figure}
	\centering
	\includegraphics[width=\linewidth]{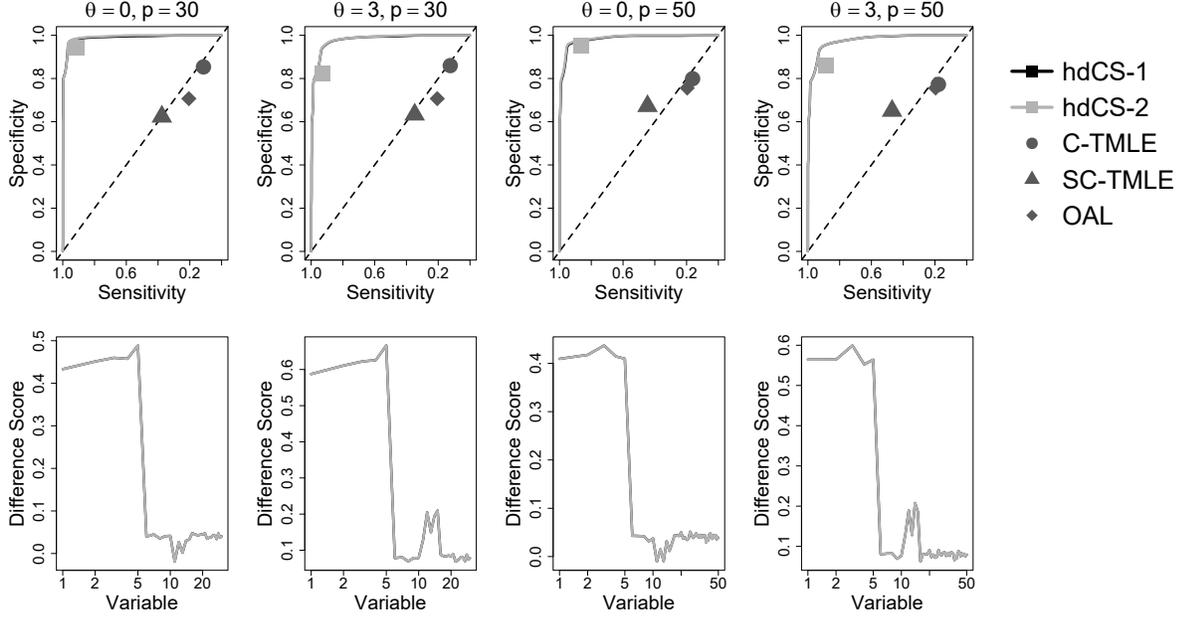}
	\caption{Results of numerical experiments with model misspecification (see label of Figure~\ref{fig:simLow}).  }
	\label{fig:simModMis}
\end{figure} 

\section{Analysis of Boston Housing Data}
\label{sec:bostonData}

In this section, we analyze the Boston housing dataset originally analyzed in \cite{harrisonjr1978hedonic} and more recently analyzed for variable importance but in the context of prediction modeling~\citep{doksum1995nonparametric,bi2003dimensionality,friedman2008predictive,williamson2017nonparametric}. The dataset is publicly available in the \texttt{R} package \texttt{MASS}~\citep{venables2002modern}. It contains the median value of owner occupied homes in the Boston metropolitan area for $n=506$ neighborhoods defined by the 1970 census tracts. We study the relationship between home value and proportion of black residents in a neighborhood. Our binary exposure variable is an indicator of proportion of black residents being less than the median proportion. We expect neighborhoods with a high proportion of black residents to have lower median home value~\citep{perry2018devaluation}. 

The dataset contains twelve variables which can be potential confounders, we divide them into four groups as identified in \cite{williamson2017nonparametric}. The first group consists of neighborhood features: the proportion of the population of lower socio-economic status~(\texttt{lstat}), referring to adults without any high school education or male workers classified as laborers; the crime rate~(\texttt{crim}); the proportion of a town’s residential land zoned for lots greater than 25,000 square feet~(\texttt{zn}); the proportion of non-retail business acres per town~(\texttt{indus}); the full-value property-tax rate~(\texttt{tax}); the pupil-teacher ratio by school district~(\texttt{ptratio}); and an indicator of whether the tract of land borders the Charles River~(\texttt{chas}). The second group consists of accessibility features: the weighted distance to five employment centers in the Boston region~(\texttt{dis}); and an index of accessibility to radial
highways~(\texttt{rad}). The third group consists of structural features: the average number
of rooms in owner units; and the proportion of owner units built prior to 1940. The final group consists
of one variable alone: the nitrogen oxide concentration, a measure of air pollution.

We implement hdCS-1, hdCS-2, C-TMLE, SC-TMLE and OAL to estimate the confounder set. We allow for flexible modeling of $Q_j$ and $\pi_j$ by using polynomial regression of degree three. For C-TMLE and SC-TMLE, we use a super-learner with linear and additive models as potential learning algorithms. For this analysis, we also implement hdCS for groups of variables defined by the four groups above. For each method, we use the selected variable set as adjustment variables for the linear model with response as the median home value and exposure as the indicator of black population (low vs. high proportion of black residents).   

The hdCS estimates of variable and group importance scores are presented in Figure~\ref{fig:bostonDat}. A noteworthy aspect of our analysis is that hdCS selected nine variables based on variable importance score but only three variables based on group scores. In Table~\ref{tab:bostonDat}, we show the estimated mean difference in median home value between the two groups~(a positive difference corresponds to a low home value in the group with  high proportion of black residents). The naive mean difference (corresponding to variables selected by C-TMLE), shows the opposite effect to what we expect, although the result is not statistically significant. Other competing methods show a positive mean difference however the estimate is not statistically significant at the $\alpha=0.05$ level. Both hdCS and group hdCS lead to statistically significant adjusted mean difference in home value. A striking feature of this analysis is how the effect size is very similar for both hdCS and group hdCS despite having very different adjustment variables.

\begin{table}
	\small\sf \centering
	\caption{Estimates of adjusted mean difference in median home value based on linear models with adjustment variables selected by hdCS and competing methods.}
	\label{tab:bostonDat}
	\bgroup
	\def\arraystretch{1.5}%
	\begin{tabular} {|l | r l | p{40mm} |}
		\hline
		\textbf{Method} & \multicolumn{2}{c|}{\textbf{Mean difference (USD)} } & \textbf{Adjustment variables} \\
		\hline
		\hline
		C-TMLE/Unadjusted &  -556\$ & ($p = 0.497$) & \texttt{-}  \\
		\hline
		SC-TMLE & 308\$ & ($p = 0.686$) & \texttt{crim}\\
		\hline
		OAL & 827\$ & ($p= 0.069$) & \texttt{chas  nox  rm  dis \newline ptratio  lstat}\\
		\hline
		\hline
		hdCS & 1882\$ & ($p<0.001$) & \texttt{crim  zn  indus  nox \newline  age  dis  rad  tax  lstat}\\
		\hline
		Group hdCS & 1853\$ & ($p=0.013$) & Air Quality~(\texttt{nox})\newline  Access~(\texttt{dis  rad})\\
		\hline
	\end{tabular}
	\egroup
\end{table}

\begin{figure}
	\centering
	\includegraphics[width=\linewidth]{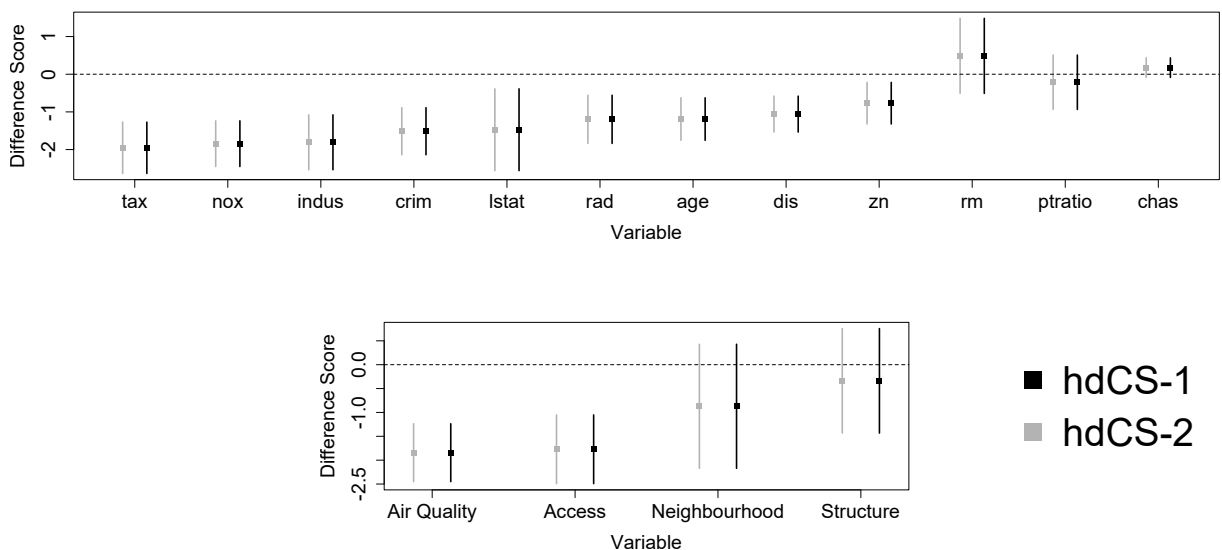}
	\caption{\emph{Top}: Estimated values of variable importance score, $\phi_j$, for $j=1,\ldots,12$ with 90\% confidence intervals for hdCS-1 and hdCS-2. Variables are sorted according to their rank determined by the absolute value of estimated importance score. \emph{Bottom}: Estimated values of group importance scores, $\Phi_j$, for $j=1,\ldots,4$ with 90\% confidence intervals. Variable groups are sorted according to the absolute value of estimated group importance score.}
	\label{fig:bostonDat}
\end{figure}

\section{Conclusion}
\label{sec:conclusion}

In this manuscript we present hdCS a novel technique for ranking and selection of confounders for high-dimensional data. We presented estimators which are efficient, asymptotically linear and doubly robust to model misspecification. Further model misspecification can be mitigated by using flexible, nonparametric modeling techniques and state-of-the-art machine learning tools. Based on results from the theory of influence functions, we can establish asymptotic normality and efficiency for our estimators. We demonstrated the practical advantages of our proposals via numerical studies on simulated and real data. 

For future work, we propose to tackle the issues and limitations of the Bross formula. Specifically, we propose a theoretical derivation of the parameters $\psi_j$ and $\phi_j$ under more general data generating schemes. The ultimate goal being, to identify necessary and sufficient conditions under which ranking based on scores $\psi_j$ or $\phi_j$ are valid. Another potential research direction is extending our methodology beyond binary exposure variables. One possible approach is to define a \emph{confounding impact curve}: $\zeta_{j}(e) = \expec \{ \tau_j(X_{ji}| E_i = e)\}$. As a special case we recover the parameters of this manuscript by noting $\psi_j = \zeta_{j}(1)/\zeta_j(0)$ and $\phi_j = \zeta_{j}(1) - \zeta_{j}(0)$. For the general case we could define a measure of confounding impact as $\max_e \zeta_{j}(e) - \min_e \zeta_j(e)$. Other functions of $\zeta_j$ could be considered, and a comparative study of various measures of confounding impact is a promising direction for future research.  

All methods presented in this paper, have been implemented in the R package \texttt{hdCS} which will soon be made available on github.

%\begin{acks}
%Acknolwdgements here
%\end{acks}

%\begin{dci}
%	The Authors declare that there is no conflict of interest.
%\end{dci}
%
%\begin{funding}
%A.H. was partially funded by the STATLAB-CRM-CANSSI postdoctoral grant. This work was partially funded by Canadian Institutes of Health Research (CIHR)[FDN-143297]; and the Canadian Statistical Sciences Institute (CANSSI) [Collaborative Research Team Project “Statistical Analysis of Large Administrative Health Databases: Emerging Challenges and Strategies”]. R.P. Holds the Albert Boehringer I Chair in Pharmacoepidemiology.
%\end{funding}
%
%\begin{sm}
%The online supplementary materials contain technical details, proofs and additional details regarding the TMLE algorithm. 
%\end{sm}

%%Vancouver (numbered)
\bibliographystyle{plainnat}
\bibliography{refnew.bib}

\appendix 

\section{Web Appendix A}
In this appendix, we begin with a proof of Lemma~1 followed by a proof of the double robustness property of our estimator  $\wh{\theta}_j^{dr}$. For ease of reading we re-state Lemma~1 here:
\begin{customthm}{1}
	%\label{lemma:propScore}
	For a random vector $(O_i, E_i, C_{1i},\ldots, C_{pi})\sim P$ the parameters 
	\[ \psi_j(P) = \frac{\expec\{\tau_j(C_{ji})|E_i=1\} }{\expec\{\tau_j(C_{ji})|E_i=0\}},
	\]
	\[  
	\phi_j(P) =  {\expec\{\tau_j(C_{ji})|E_{i}=1\} }- {\expec\{\tau_j(C_{ji})|E_i=0\}},
	\] can equivalently be written as 
	\[\psi_j(P) = \frac{\expec\{O_i\pi_j(C_{ji})\}/\pr(E_i=1) }{\expec[O_i\{1 - \pi_j(C_{ji})\}]/\pr(E_i=0) },\] and 
	\[\phi_j(P) =  {\frac{\expec\{O_i\pi_j(C_{ji})\}}{\pr(E_i=1)} } - {\frac{\expec[O_i\{1 - \pi_j(C_{ji})\}]}{\pr(E_i=0)} }.\]	
\end{customthm}
\begin{proof}
	The proof of this follows immediately by definition of conditional expectations. We have 
	\begin{align*}
		\expec \{ \tau_j(C_{ji}) |E_i=1 \} &= \expec \left\{ \frac{\tau_j(C_{ji})I(E_i =1)}{\pr(E_i=1)} \right\} \\
		&= \expec   \left[ \expec \left\{ \frac{\tau_j(C_{ji})I(E_i =1)}{\pr(E_i=1)} \right\} \Big| C_{ji} \right] \\
		&= \expec    \left[ \frac{\tau_j(C_{ji})\expec\{ I(E_i =1) |C_{ji} \}}{\pr(E_i=1)} \right] \\
		&= \expec    \left[ \frac{\tau_j(C_{ji})\pr\{ E_i =1 |C_{ji} \}}{\pr(E_i=1)} \right] \\
		&= \expec    \left[ \frac{\expec(O_{i}|C_{ji}) \pi_j(C_{ji}) }{\pr(E_i=1)} \right] \\
		&= \expec    \left[ \expec \left\{ \frac{O_i \pi_j(C_{ji}) }{\pr(E_i=1)}\right\} \Big| C_{ji}\right] \\
		&=   \frac{ \expec \left\{ O_i \pi_j(C_{ji}) \right\} }{\pr(E_i=1)}.
	\end{align*}
	By an identical argument we can show that 
	\[\expec \{ \tau_j(C_{ji}) |E_i=0 \}  =  \frac{ \expec \left[ O_i \{1- \pi_j(C_{ji})\} \right] }{\pr(E_i=0)}, \]
	which completes the proof.
\end{proof}

We now prove a main result regarding our one-step correction estimator or the doubly robust estimator $\wh{\theta}_j^{dr}$. Recall that our one step correction is given by: 
\begin{align}
	\wh{\theta}^{dr}_{j,n} %&= \theta_{j}(\wh{\pi}_{j,n}^{dr},\wh{\tau}_{j,n}^{dr}) + n^{-1}\sum_{i=1}^{n} D_{\theta_{j}(\wh{\pi}^{dr}_{j,n}, \wh{\tau}_{j,n}^{dr})} \nonumber\\
	&=  \wh{\theta}^{naive}_{j,n} +  \frac{1}{n}\sum_{i=1}^{n}  O_i \wh{\pi}_{j,n}(C_{ji}) - \wh{\tau}_{j,n}(C_{ji}) \wh{\pi}_{j,n}(C_{ji})\nonumber \\
	&=  \frac{1}{n} \sum_{i=1}^{n} I(E_i=1) \wh{\tau}_{j,n}(C_{ji})  +  O_i \wh{\pi}_{j,n}(C_{ji}) - \wh{\tau}_{j,n}(C_{ji}) \wh{\pi}_{j,n}(C_{ji}). \label{eqn:DRestimate2}
\end{align}

\begin{lemma}
	\label{lemma:DRestimate}
	For a sequence of point-wise convergent estimates, i.e. $\wh{\tau}^{dr}_{j,n} \to \tau^*_j$, and $\wh{\pi}^{dr}_{j,n} \to \pi^*_j$, the estimator (\ref{eqn:DRestimate2}) is doubly robust in the following sense: if either $\tau^*_j(\cdot) = \tau_j(\cdot; P^0)$ or $\pi^*_j(\cdot) =\pi_j(\cdot; P^0)$ then $\wh{\theta}_{j,n}^{dr} \to \theta_j(P^0)$.
\end{lemma} 
\begin{proof}
	By the WLLN and our convergence assumption we have 
	\begin{align*}
		\wh{\theta}_j^{dr} \to \expec_{P^0}\left[ O_i {\pi}^*_j(C_{ji}) + {\tau}^*_j(C_{ji}) \left\{ I(E_i=1) - {\pi}^*_j(C_{ji}) \right\}  \right].
	\end{align*}
	First assume that $\tau^*_j(\cdot) = \tau_j(\cdot; P^0)$. Then
	\begin{align*}
		\wh{\theta}_j^{dr} &\to \expec_{P^0}\left[ {\pi}^*_j(C_{ji})\{O_i - \tau_j(C_{ji};P^0)\} + I(E_i=1){\tau}_j(C_{ji};P^0)  \right]\\
		&= \expec_{P^0}\left[ {\pi}^*_j(C_{ji})\{O_i - \tau_j(C_{ji};P^0)\}\right] + \theta_j(P^0)\\
		&= \expec_{P^0}\left[ {\pi}^*_j(C_{ji})\{\expec_{P^0}(O_i|C_{ji}) - \tau_j(C_{ji};P^0)\}\right] + \theta_j(P^0)\\
		&= \theta_j(P^0),
	\end{align*} 
	since by definition $\expec_{P^0}(O_i|C_{ji}) = \tau_j(C_{ji};P^0)$. Assume instead that $\pi^*_j(\cdot) = \pi_j(\cdot; P^0)$, then
	\begin{align*}
		\wh{\theta}_j^{dr} &\to \expec_{P^0}\left[ O_i{\pi}_j(C_{ji};P^0) + {\tau}^*_j(C_{ji})\{I(E_i=1) - \pi_j(C_{ji}; P^0)\}  \right]\\
		&= \expec_{P^0}\left\{ O_i{\pi}_j(C_{ji};P^0)\right\} + \expec_{P^0}\left[{\tau}^*_j(C_{ji})\{ \expec_{P^0}\{I(E_i=1)|C_{ji}\} - \pi_j(C_{ji}; P^0)\}  \right]\\
		&= \theta_j(P^0),
	\end{align*}
	where $\expec_{P^0}\left\{ O_i{\pi}_j(C_{ji};P^0)\right\} = \theta_j(P^0)$ by the proof of Lemma~\ref{lemma:propScore} and $\expec_{P^0}\{I(E_i=1)|C_{ji}\} =\pr(E_i=1|C_{ji}) = \pi_j(C_{ji}; P^0)$ by definition.
\end{proof}

\section{Web Appendix B}
In this appendix, we present details for the targeted maximum likelihood estimator~(TMLE), $\wh{\theta}_j^{tl}$ and also theoretical guarantees for our efficient estimators. In the first section, we will prove Theorem~1 after introducing appropriate background and accompanying results. In the second section of this appendix, we will derive the full TMLE algorithm and provide some extensions/variations. 

\subsection{Proof of Theorem 1}

We begin with some notations and definitions. For i.i.d. data $\bs{X}_i = (O_i, E_i, C_{1i}, \ldots, C_{pi})\sim P \ (i=1,\ldots,n)$  for some probability measure $P$, recall the definitions from the main manuscript: 
\begin{equation}
	\label{eqn:simpleParas}
	\theta_j(P) = \expec_P\{I(E_i=1) \tau_j(C_{ji}; P ) \}, \quad \mu_O(P) = \expec_P(O_i), \quad \mu_E(P) = \expec_P(E_i),
\end{equation}
where $\tau_j(c; P) = \expec_P(O_i|C_{ji} = c)$. Also recall that our variable importance parameters can be written in terms of the parameters~\eqref{eqn:simpleParas}: 
\begin{align*}
	%\label{eqn:poiBySimple}
	\psi_j(P) &= \frac{\theta_j(P)/\mu_E(P)}{\{ \mu_O(P) - \theta_j(P) \}/\{ 1 - \mu_E(P) \}},\\
	\phi_j(P) &= {\theta_j(P)/\mu_E(P)} - {\{ \mu_O(P) - \theta_j(P) \}/\{ 1 - \mu_E(P) \}}.
\end{align*}
As before, we denote the propensity score by $\pi_j(c; P) = \pr(E=1|C_j=c)$. Recall the \emph{adjusted exposure-response model} defined as $Q_{j}(e, c_j; P) = \expec_P(O|E=e, C_j = c_j)$. Also recall the relationship: 
\begin{equation}
	\label{eqn:3func}
	\tau_j(c_j;P) = \pi_j(c_j;P)Q_j(1, c_j;P) + \{1-\pi_j(c_j;P)\}Q_j(0,c_j;P).
\end{equation}

For any $P$-measurable function $f$, we define $Pf = \int f(\bs{x})\, dP(\bs{x})$ and $P_nf = n^{-1}\sum_{i=1}^{n}f(\bs{X}_i)$.  We define a \emph{statistical model} as a collection of probability measures denoted by $\mathcal{M}$, for our manuscript (and this appendix), we only consider nonparametric models.

We now prove that our estimators are asymptotically linear, this is done using a functional version of the Taylor's theorem where recall the first derivative is the \emph{influence function}. We first formally define an influence function.
\begin{definition}
	\label{def:IC}
	Consider a functional $\psi: \mathcal{M} \to \mathbb{R}$. Let $p$ denote the Radon-Nikodym derivative with respect to some dominating measure $\nu$, corresponding to the measure $P\in \mathcal{M}$. Let $\{P_{\e}: \e \in \mathbb{R} \}$ be a one dimensional paramteric sub-model such that at $\e = 0$,  $P_{0} = P$ and the score function is given by $h(\cdot)$. Then the \emph{influence curve}, $D_{\psi(P)}$, is a function satisfying $PD_{\psi(P)} = 0$ for any $P\in \mathcal{M}$ and the following relationship:
	\begin{equation}
		\label{eqn:EICdef}
		\left. \frac{\partial }{\partial \e} \psi(P_{\e}) \right|_{\e = 0} = P\left( D_{\psi(P)}h \right) = \int D_{\psi(P)}(\bs{x})h(\bs{x}) dP(\bs{x}).
	\end{equation}
	In the case of nonparametric models $\mathcal{M}$, there is only one derivative and that is called the efficient influence curve~(EIC).
\end{definition}

\begin{lemma}
	\label{lemma:EIC}
	%The parameters $\theta_j(P), \mu_O(P)$ and $\mu_E(P)$ are pathwise differentiable at each $P\in \mathcal{M}$ relative to $\mathcal{M}$. 
	For a random vector $\bs{X} = (O, E, C_1,\ldots, C_p)\sim P$, the efficient influence functions are defined as
	\begin{align*}
		%\label{eqn:EICofPOI}
		D_{\theta_j(P)}(\bs{x}) &= o\pi_j(c_j; P) + \tau_j(c_j;P)\{I(e=1) - \pi_j(c_j;P)\} - \theta_j(P),\\
		D_{\mu_O(P)}(\bs{x}) &= o - \mu_O(P),\\ 
		D_{\mu_E(P)}(\bs{x}) &= I(e=1) - \mu_E(P).
	\end{align*}
\end{lemma}  
\begin{proof}
	We will show that the above functions satisfy Definition~\ref{def:IC}. We begin with 
	\begin{align*}
		\theta_{j}(P_{\e}) &= \int I(e=1) \tau_j(c; P_{\e}) dP_{\e}(e,c).
	\end{align*}
	For convenience we drop the subscript $j$ from $c_j$. Taking the derivative at $\e=0$ gives us the following by the chain product rule: 
	\begin{align}
		\left. \frac{\partial \theta_{j}(P_{\e}) }{\partial \e} \right|_{\e=0} &= \int I(e=1)\tau_j(c; P)h(c,e)dP(c,e)\label{eqn:line1} \\
		&+ \int I(e=1)  \left. \frac{\partial \tau_{j}(c;P_{\e}) }{\partial \e} \right|_{\e=0} dP(c,e)\label{eqn:line2},
	\end{align}
	where by definition $h = dP'/dP$ (i.e. the score function is the first derivative of the log-likelihood). 
	
	For the second term we note that 
	\begin{align*}
		\left. \frac{\partial \tau_{j}(c;P_{\e}) }{\partial \e} \right|_{\e=0} &= \left. \frac{\partial }{\partial \e}   \expec_{P_{\e}} (O|C_j)\right|_{\e=0} \\
		&= \left. \int o \frac{\partial p_{\e}(o|c)}{\partial \e} \right|_{\e=0} d\nu(o),\\
		&= \int o h(o|c) p(o|c)d\nu(o)
	\end{align*}
	where (with some abuse of notation) $p(o|c)$ is the conditional density (Radon-Nikodym derivative) with respect to some domintaing measure $\nu$, and $h(o|c)$ is the corresponding score function.
	
	Now note that \eqref{eqn:line1} is already in the form $\int D_{\psi(P)} h dP$, for \eqref{eqn:line2} we have the following:
	\begin{align*}
		\eqref{eqn:line2} &= \int I(e=1) \left\{ \int o h(o|c) p(o|c)d\nu(o) \right\} dP(e,c)\\
		&= \int I(e=1) o h(o|c) \frac{p(o,c)}{p(c)}p(e,c) d\nu(e,c,o)\\
		&= \int I(e=1) p(e|c) o h(o|c) p(o,c) d\nu(e,c,o)\\
		&= \int \pi_j(c; P)\times  o\times  h(o|c) p(o,c) d\nu(c,o).
	\end{align*}
	Now by the definition of score function, and properties of conditional probabilities, $h(o,c) = h(o|c) + h(c) \Rightarrow h(o|c) = h(o,c) - h(c)$. Thus, continuing our derivation:
	\begin{align*}
		\eqref{eqn:line2} &= \int \pi_j(c; P)\times  o\times  \{ h(o,c) - h(c) \} p(o,c) d\nu(c,o)\\
		&= \int \pi_j(c;P)\times o \times h(o,c)dP(o,c) - \int \pi_j(c;P) \times o \times h(c) p(o|c)p(c) d\nu (c,o)\\
		&= \int o\pi_j(c;P) h(o,c)dP(o,c) - \int \pi_j(c;P) \left\{  \int o p(o|c)d\nu(o)\right\} h(c) p(c) d\nu (c)\\
		&= \int o\pi_j(c;P) h(o,c)dP(o,c) - \int \pi_j(c;P)\tau_j(c;P) h(c) dP(c). 
	\end{align*}
	
	By the properties of the score function we have that for any function $f$, 
	\begin{align*}
		\int f(c) h(o,c) p(o,c) d\nu(o,c) &= \int f(c) h(c) dP(c). 
	\end{align*}
	Thus we can collect the three different integrals into one:
	\begin{align*}
		%\label{eqn:thetajDeriv}
		\left. \frac{\partial \tau_{j}(c;P_{\e}) }{\partial \e} \right|_{\e=0} &= \int \left\{ I(e=1)\tau_j(c;P) + o\pi_j(c;P) - \pi_j(c;P)\tau_j(c;P)\right\} h(o,c,e) dP(o,c,e).
	\end{align*}
	Thus we have $D^{\diamond}_{\theta_j(P)} = o \pi_j + \tau_j\{I(e=1) - \pi_j\}$ satisfying \eqref{eqn:EICdef}, for this to be the EIC we just need to recenter, i.e. $D_{\theta_{j}(P)} = D^{\diamond}_{\theta_j(P)} - PD^{\diamond}_{\theta_j(P)}$. The proof is completed by noting that $PD^{\diamond}_{\theta_j(P)} = \theta_j$. EIC for the other parameters is an easy calculation: 
	\begin{align*}
		\left. \frac{\partial \mu_O(P_{\e}) }{\partial \e} \right|_{\e=0} &= \left. \frac{\partial  }{\partial \e} \int o dP_{\e}(o)  \right|_{\e=0}\\
		&= \int o  h(o) dP(o),
	\end{align*}
	and again by centering we achieve $D_{\mu_{O}(P)}(\bs{x}) = o - \mu_{O}(P)$. The result is identical for $\mu_E(P)$. 
	
\end{proof}

\begin{lemma}
	\label{lemma:Expansion}
	For a nonparametric statistical model $\mathcal{M}$ and i.i.d. data $\bs{X}_i = (O_i, E_i, C_{1i}, \ldots, C_{pi})\sim P^0\in \mathcal{M}$, we have the following expansions for all $P\in \mathcal{M}$:
	\begin{align*}
		\theta_j(P) - \theta_j(P^0) &= P^0_nD_{\theta_j(P^0)} - P^0_nD_{\theta_j(P)} + (P^0_n - P^0)\left\{D_{\theta_j(P)} - D_{\theta_j(P^0)}  \right\} + R_{\theta_j}(P, P^0),\\
		\mu_O(P) - \mu_O(P^0) &= P^0_nD_{\mu_O(P^0)} - P^0_nD_{\mu_O(P)},\\
		\mu_E(P) - \mu_E(P^0) &= P^0_nD_{\mu_E(P^0)} - P^0_nD_{\mu_E(P)},
	\end{align*}
	where the remainder term is defined as
	\begin{align*}
		R_{\theta_j}(P, P^0) &= P^0 \left[ \{\tau_j(c_j; P^0) - \tau_j(c_j; P) \}\{ \pi_j(c_j;P) - \pi_j(c_j;P^0) \} \right].
	\end{align*}
\end{lemma} 

\begin{proof}
	Note that the expansion for $\mu_{O}$ and $\mu_E$ follows immediately by definition. For $\theta_j$ we begin with the following representation:
	\begin{align}
		\label{eqn:thetaExpBasic}
		\theta_j(P) - \theta_j(P^0) &= - P_n^0D_{\theta_{j}(P)} + R_{\theta_j}(P,P^0),
	\end{align}
	where 
	
	\begin{align*}
		R_{\theta_j}(P,P^0) &= \theta_j(P) - \theta_j(P^0) + P^0 D_{\theta_{j}}(P)\\
		&= \theta_{j}(P) - P^0\left\{ I(E_i=1)\tau_j(C_j; P^0) \right\} \\
		&+ P^0\left[ O_i\pi_j(C_j;P) + \tau_j(C_j; P)\left\{ I(E_i = 1) - \pi_j(C_j; P) \right\} \right] - \theta_{j}(P)\\
		&= \expec_{P^0} \left[ \pi_j(C_j;P)\{O_i - \tau_j(C_j;P)\} - I(E_i = 1)\{\tau_j(C_j;P^0) - \tau_j(C_j; P) \} \right]\\
		&= \expec_{P^0} \left[ \pi_j(C_j;P)\{\expec_{P^0} (O_i|C_j) - \tau_j(C_j;P)\} - \expec_{P^0}\{I(E_i = 1)|C_j\}\{\tau_j(C_j;P^0) - \tau_j(C_j; P) \} \right]\\
		&= \expec_{P^0} \left[ \pi_j(C_j;P)\{\tau_j(C_j; P^0) - \tau_j(C_j;P)\} - \pi_j(C_j;P^0)\{\tau_j(C_j;P^0) - \tau_j(C_j; P) \} \right]\\
		&= \expec_{P^0}\left[ \left\{ \tau_j(C_j; P^0) - \tau_{j}(C_j; P) \right\} \left\{ \pi_{j}(C_j; P) - \pi_{j}(C_j; P^0) \right\}\right]. 
	\end{align*}
	With that, the full expansion of $\theta_j(P) - \theta_{j}(P^0)$ is straightforward manipulation: 
	\begin{align*}
		\theta_j(P) - \theta_j(P^0) &= -P^0D_{\theta_{j}(P)} + R_{\theta_j}(P,P^0)\\
		&= (P^0_n-P^0)D_{\theta_{j}(P)} - P^0_nD_{\theta_{j}(P)} + R_{\theta_j}(P,P^0)\\
		&= (P^0_n-P^0)\left\{D_{\theta_{j}(P)} - D_{\theta_{j}(P^0)}  \right\} +(P^0_n - P^0)D_{\theta_{j}(P^0)}  - P^0_nD_{\theta_{j}(P)} + R_{\theta_j}(P,P^0)\\
		&= P^0_nD_{\theta_{j}(P^0)}  - P^0_nD_{\theta_{j}(P)} + (P^0_n-P^0)\left\{D_{\theta_{j}(P)} - D_{\theta_{j}(P^0)}  \right\} + R_{\theta_j}(P,P^0),
	\end{align*} 
	where the last equality follows from $PD_{\psi(P)} = 0$ for any probability measure $P$. 
\end{proof}

The expansion of the previous part can be broken into four parts:
\begin{align*}
	\underbrace{P^0_nD_{\theta_j(P^0)} }_{I} - \underbrace{P^0_nD_{\theta_j(P)}}_{II} + \underbrace{(P^0_n - P^0)\left\{D_{\theta_j(P)} - D_{\theta_j(P^0)}  \right\} }_{III} +  \underbrace{R_{\theta_j}(P, P^0)}_{IV}.
\end{align*}

Now (I) is a sample mean of some function of our data, thus we can apply the CLT on this part. The term (II) can be dealt with in two ways corresponding to our two estimators. The terms (III) and (IV) are called the \emph{empirical process term} and \emph{remainder}, respectively. These last two terms are asymptotically negligible under conditions~\ref{assum:slowConvergence} and \ref{assum:Donsker} stated in the main manuscript.  
%\begin{condition}
%	\label{assum:slowConvergence}
%	Consider a sequence of estimators $(\wh{\tau}_j^n)_{n=1}^{\infty}$ and $(\wh{\pi}_j^n)_{n=1}^{\infty}$ we have, $\|\wh{\tau}^n_j(c) - \tau_j(c;P^0)\|_{\infty} = o_p(n^{-1/4})$ and $\|\wh{\pi}^n_j(c) - \pi_j(c;P^0)\|_{\infty} = o_p(n^{-1/4})$.
%\end{condition}
%\begin{condition}
%	\label{assum:Donsker}
%	With initial estimators $\wh{\tau}_j^n$ and $\wh{\pi}_j^n$, the function $\bs{x} \mapsto D_{\theta_j(\wh{\tau}_j^n, \wh{\pi}_j^n)}(\bs{x})$ is $P^0$--Donsker~\citep{vaart1998asymptotic}.
%\end{condition}

These conditions substantially simplify our functional expansions. We do this in the following lemma. For notational convenience, we write $R_{\theta_j}(P,P^0) \equiv R_{\theta_j}[\{\pi_j(;P), \tau_j(;P)\},P^0]$ and $\theta_j(P) = \theta_j\{\pi_j(;P), \tau_j(;P)\}$ to clarify dependence on $\pi_j$ and $\tau_j$. 

\begin{lemma}
	\label{lemma:AsympcNeg}
	Under Conditions \ref{assum:slowConvergence} and \ref{assum:Donsker}, 
	\begin{equation*}
		R_{\theta_j}\{(\wh{\pi}_j^n,\wh{\tau}_j^n),P^0\} = o_p(n^{-1/2}), \quad (P^0_n - P^0)\left\{D_{\theta_j(\wh{\tau}_j^n, \wh{\pi}_j^n)} - D_{\theta_j(P^0)}  \right\} = o_p(n^{-1/2}),
	\end{equation*}
	and consequently our Taylor expansion reduces to:
	\begin{equation}
		\label{eqn:mainExpansion}
		\left[ \begin{array}{c}
			\wh{\Theta}^n_j\\% + P^0_nD_{\theta_j(\wh{P}_n)} \\
			\overline{O}\\
			\overline{E}\\
		\end{array} \right] - 		\left[ \begin{array}{c}
			{\theta}_j(P^0) \\
			\mu_O(P^0)\\
			\mu_E(P^0)\\
		\end{array} \right] =  		\left[ \begin{array}{c}
			P^0_nD_{\theta_j(P^0)} \\
			P^0_nD_{\mu_O(P^0)}\\
			P^0_nD_{\mu_E(P^0)}\\
		\end{array} \right] + o_p(n^{-1/2}),
	\end{equation}
	where $\wh{\Theta}^n_j = \theta_j(\wh{\pi}_j^n,\wh{\tau}_j^n) + P^0_nD_{\theta_j(\wh{\pi}_j^n,\wh{\tau}_j^n)}$, $\bar{O} = n^{-1}\sum_{i=1}^{n}O_i$ and $\bar{E} = n^{-1}\sum_{i=1}^{n}E_i$.
\end{lemma}

\begin{proof}
	For the remainder term, it immediately follows that it is $o_p(n^{-1/2})$ since it is a product of two terms that are $o_p(n^{-1/4})$. For the empirical process term, we use standard results in empirical process theory under Condition~\ref{assum:Donsker}~(see e.g. Lemma~19.24 of \cite{vaart1998asymptotic}). These results combined with Lemma~\ref{lemma:Expansion} prove the expansion \eqref{eqn:mainExpansion}.
\end{proof}

\begin{proof}[Proof of Theorem 1]
	Proof of Theorem~1 follows simply by considering the first line in the expansion of Lemma~\ref{lemma:AsympcNeg}. Recall that for any initial estimator $\wh{\tau}_{j,n}$ and $\wh{\pi}_{j,n}$, our doubly robust estimator is given by 
	\begin{equation*}
		\wh{\theta}_j^{dr} = \theta_j(\wh{\pi}_{j,n},\wh{\tau}_{j,n}) + P^0_nD_{\theta_j(\wh{\pi}_{j,n},\wh{\tau}_{j,n})}.
	\end{equation*}
	For $\wh{\theta}_{j}^{tl}$ the TMLE algorithm ensures that starting with any initial estimators satisfying Conditions~\ref{assum:slowConvergence} and \ref{assum:Donsker}, the final TMLE estimators $\wh{\tau}_{j}^{k+1}$ and $\wh{\pi}_j^{k+1}$ (where $k$ denotes the last TMLE iterate) will also satisfy Conditions~\ref{assum:slowConvergence} and \ref{assum:Donsker}~\citep{laan2011targeted}. The TMLE procedure ensures that $P^0_nD_{\theta_j(\wh{\pi}_j^{k+1},\wh{\tau}_j^{k+1})} = 0$. The remainder of this expansion remains unchanged and that completes the proof. 
\end{proof}

\subsection{Details of TMLE algorithm}
Our TMLE algorithm was derived using the construction detailed in Chapter 5 of \cite{laan2011targeted}. For the sake of completeness we will outline the TMLE algorithm for a general parameter and discuss how this leads to the derivation of our main algorithm. As before, say we have data $\bs{X}_i\sim P^0 \{i=1,\ldots,n\}$. 

\begin{enumerate}
	\item For a target parameter $\psi : \mathcal{M}\to \mathbb{R}$, find the efficient influence curve~(EIC) denoted by $D_{\psi(P)}$. Often times the target parameters depends on $P$ through functions of $P$: $\psi(P) = \psi(Q_1(P), Q_2(P), Q_3(P))$ where $(Q_1,Q_2,Q_3)$ are variation-independent. In this case, we can decompose the EIC as 
	\begin{align*}
		D_{\psi(P)} &= D_{\psi(Q_1)} + D_{\psi(Q_3)} + D_{\psi(Q_3)}. 
	\end{align*}
	This decomposition can be achieved by projecting $D_{\psi(P)}$ on appropriate \emph{tangent spaces}~(for details see Chapter 5, \cite{laan2011targeted}). 
	%For our parameter $\theta_j$ this was defined in Lemma~\ref{lemma:EIC}.
	\item Define a loss function $L(\cdot)$ such that 
	\begin{align*}
		P^0 \gets \underset{P\in \mathcal{M}}{\operatorname{argmin } }\ \expec_{P^0}L(P).
	\end{align*}
	If $\psi(P) = \psi(Q_1(P), Q_2(P), Q_3(P))$ for variation-independent $(Q_1,Q_2,Q_3)$,
	% In our case, our parameter $\theta_j(P)$ depends on $Q_j(\cdot,\cdot; P)$, $\pi_j(\cdot; Q)$ and $P_{C_j}$~(the marginal distribution of $C_j$). These parts are variation-independent by the decomposition $p(o,e,c_j)  = p(o|e,c_j)p(e|c_j)p(c_j)$ where $p(\cdot)$ are density functions~(or Radon-Nikodym derivatives). In the case of variation independent parts 
	then define loss functions $L_1, L_2, L_3$ such that:
	\begin{align*}
		Q_j(P^0) \gets \underset{Q_j }{\operatorname{argmin } }\ \expec_{P^0}L_j(Q_j).
	\end{align*}
	for $j=1,2,3$. 
	\item Define a parametric working model $\{P_{\e}: \e \in I\}$ for some open interval $I\subset \mathbb{R}$, such that $P_{0} = P$ and satisfies
	\begin{equation*}
		\left. \frac{d}{d\e} L(P_{\e}) \right|_{\e=0} = D_{\psi (P)}. 
	\end{equation*}
	For the variation-independent decomposition, define a parametric working model for each part: $Q_{1,\e_1}$, $Q_{2,\e_2}$ and $Q_{3,\e_3}$ such that for $j\in\{1,2,3\}$,  $Q_{j,0} = Q_{j}$, and 
	\begin{equation*}
		\left. \frac{d}{d \e_j} L_j(Q_{j,\e_j}) \right|_{\e_j=0} = D_{\psi (Q_j)}. 
	\end{equation*}
	\item Given an initial estimator $\wh{P}^0$ of $P^0$, compute the update 
	\begin{align*}
		\e^0 \gets \underset{\e }{\operatorname{argmin } } \sum_{i=1}^{n} L(\wh{P}^0_{\e})(\bs{X}_i). 
	\end{align*}
	In case of the variation-independent decomposition, given initial estimators $\wh{Q}_j^0$ of $Q_j(P^0)$, compute the updates for all $j$
	\begin{align*}
		\e_j^0 \gets \underset{\e}{\operatorname{argmin } } \sum_{i=1}^{n} L_j(\wh{Q}^0_{j,\e})(\bs{X}_i). 
	\end{align*}
	Update the initial estimator as $\wh{P}^1 \gets \wh{P}^0_{\e^0}$, or $\wh{Q}^1_j \gets \wh{Q}_{j,\e^0}$. 
	\item Iterate this process: at the $k$-th iterate say we have $\wh{P}^k$~(or $\wh{Q}_j^{k}$). Then compute 
	\begin{align*}
		\e^k \gets \underset{\e }{\operatorname{argmin } } \sum_{i=1}^{n} L(\wh{P}^k_{\e})(\bs{X}_i), 
	\end{align*}
	and update $\wh{P}^{k+1} = \wh{P}^{k}_{\e^k}$. Similarly we can update $\wh{Q}_j^{k+1} = \wh{Q}^{k}_{j,\e_j^k}$ for $j=1,2,3$.
	\item Iterate until $\e^k = 0$ or $\e_j^k = 0$ for all $j$, resulting in the final estimator $\wh{P}^K$ or $\wh{Q}_j^{K}$. The TMLE is then the substitution estimator $\psi(\wh{P}^K)$ or $\psi(\wh{Q}_1^K, \wh{Q}_2^K, \wh{Q}_3^K)$.
\end{enumerate}
\vspace{5mm}
\noindent \textbf{We now discuss the above steps in the context of our specific parameter $\theta_j$.} 

\begin{enumerate}
	\item In our case we have $\theta_j(P) = \theta_j(Q_j, \pi_j, P_{C_j})$ where the variation-independence of our constituent parts is given by the decomposition $P = P_{O|E,C_j}P_{E|C_j}P_{C_j}$. We can decompose our EIC for $\theta_j$ as 
	\begin{align*}
		D_{\theta_j(P)}(\bs{x}) = D_{O}(\bs{x}) + D_{E}(\bs{x}) + D_{C}(\bs{x}),
	\end{align*} 
	where 
	\begin{align*}
		D_{O}(\bs{x}) &= o\pi_{j}(c_j; P) - Q_j(e,c_j;P)\pi_j(c_j;P),\\
		D_{E}(\bs{x}) &= Q_j(e,c_j;P)\pi_j(c_j;P) + \tau_j(c_j;P)\{ I(e=1) - \pi_j(c_j;P) \} - \tau_j(c_j;P)\pi_j(c_j;P) ,\\
		D_{C}(\bs{x}) &= \tau_j(c_j;P)\pi_j(c_j;P) - \theta_j(P).
	\end{align*}
	
	\item Now for the loss functions we have 
	\begin{align*}
		L_1(Q_j)(\bs{x}) &= \frac{1}{2}\{y - Q_j(e,c_j)\}^2,\\
		L_2(\pi_j)(\bs{x}) &= -e \log\{\pi_j(c_j)\} - (1-e )\log\{1 - \pi_j(c_j)\},\\
		L_3(P_{C_j})(\bs{x}) &= - \log P_{C_j}(c_j).
	\end{align*}
	
	\item We can now define the parametric working models.
	\begin{align*}
		Q_{j,\e_1} &= Q_{j} + \e_1 H_1, \\
		\text{ where } &H_1(e,c_j;P) = -\pi_j(c_j;P).\\
		\pi_{j,\e_2} &=  \expit\left\{ \logit(\pi_j) + \e_2 H_2\right\},\\
		\text{ where } &H_2(c_j;P) = -2\pi_j(c_j;P)\{Q_j(1,c_j;P) - Q_j(0, c_j;P)\} - Q_j(0,c_j;P).\\
		P_{C_j,\e_3} &= (1+\e_3H_3)P_{C_j},\\
		\text{ where } &H_3(c_j;P) = -\tau_j(c_j;P)\pi_j(c_j;P) + \theta_j(P). 
	\end{align*}
	
	\item With the above definitions of the loss function and functions $H_j$, this step immediately leads to Step 2 of the main algorithm in the manuscript. For $P_{C_j}$, let our initial estimator, $P_{n,C_j}$, be the empirical distribution. We then have that $\e_3^0 = 0$ which is obvious from the fact that 
	\begin{align*}
		\frac{d}{d\e} L_3(P_{C_j,\e}) &= -\frac{H_3 P_{C_j}}{P_{C_j} + \e H_3P_{C_j}}.\\
		\frac{d}{d\e} n^{-1}\sum_{i=1}^{n} L_3(P_{C_j,\e=0})(C_{ji}) &= n^{-1}\sum_{i=1}^{n} -H_3(C_{ji}; P_{n, C_j}),\\
		&= P_n \wh{\tau}^0_{j}(C_j)\wh{\pi}^0_j(C_j) - \expec_{P_{n,C_j}}\left\{  \wh{\tau}^0_{j}(C_j)\wh{\pi}^0_j(C_j) \right\},\\
		&= 0.
	\end{align*}
	Thus through out the iterations $\e_3^k = 0$ for all $k$.
	\item Again, this gives us Step 2 of the main algorithm in the manuscript. The initial estimators for $P_{C_j}$ do not change if we start with the empirical distribution. 
	\item This gives us the final step of our TMLE algorithm. 
\end{enumerate}

\begin{note}
	The last thing we discuss is possible variations of our original algorithm. The path $Q_{j,\e_1}$ and correcponding loss function was made for continuous response variable $O_i$. We could alternatively consider the scenario where $O_i\in \{0,1\}$ or equivalently the case of bounded continuous $O_i$ as they can be scaled to the unit interval. In this case we can define 
	\begin{align*}
		L_1(Q_j)(\bs{x}) &= -o \log\{Q_j(e,c_j)\} - (1-o )\log\{1 - Q_j(e,c_j)\},\\
		\logit(Q_{j,\e_1}) &= \logit(Q_j) + \e_1 H_1,
	\end{align*}
	for some function $H_1$. To find the function $H_1$, we need to consider
	\begin{align*}
		\frac{d}{d\e_1} L_1(Q_{j,\e_1=0})(\bs{x}) &= -H_1(e,c_j;P)(o - Q_j(e,c_j;P)),\\
		&= D_{O}(\bs{x})\ \ \  \{\text{if } H_1(e,c_j;P) = -\pi_j(c_j;P)\}.
	\end{align*}
	Similarly, estimating $\theta_j$ for other data types of $O_i$~(e.g. count data, continuous positive) is simply a matter of defining the right loss function and working parametric model. 
	
\end{note}

\end{document}